\documentclass[11pt,a4paper]{article}
\usepackage{amssymb}
\usepackage{amsmath, amscd}
\usepackage{amsfonts}
\usepackage{graphicx}
\usepackage{a4wide}
\usepackage{microtype}
\usepackage{bbm}
\usepackage{slashed}
\usepackage{times}
\usepackage{amsthm}
\usepackage{mathrsfs}
\usepackage{hyperref}
\usepackage{color}

\definecolor{darkred}{rgb}{0.8,0.1,0.1}
\hypersetup{
     colorlinks=true,         
     linkcolor=darkred,
     citecolor=blue,
}

\usepackage{comment}
\usepackage[margin=2.2cm]{geometry}
\usepackage[all,cmtip]{xy}

\usepackage{marvosym}

\theoremstyle{plain}
\newtheorem{theo}{Theorem}[section]
\newtheorem{lem}[theo]{Lemma}
\newtheorem{propo}[theo]{Proposition}
\newtheorem{cor}[theo]{Corollary}

\theoremstyle{definition}
\newtheorem{defi}[theo]{Definition}

\newtheorem{rem}[theo]{Remark}

\numberwithin{equation}{section}

\def\Vec{\mathsf{Vec}}

\def\PrBu{\mathsf{PrBuGlobHyp}}

\def\PhaseSpace{\mathfrak{PS}}
\def\PhaseSpaceOff{\mathfrak{PSOff}}

\def\CCR{\mathfrak{CCR}}

\def\astAlg{{^\ast}\mathsf{Alg}}
\def\BastAlg{B^\ast\mathsf{Alg}}
\def\CastAlg{C^\ast\mathsf{Alg}}
\def\QFT{\mathfrak{A}}
\def\PAG{\mathsf{PAG}}
\def\SubGr{\mathfrak{Q}}
\def\MWfunc{\mathfrak{MW}}

\def\SubM{\mathsf{Sub}_{\widehat{\Xi}}}
\def\Pull{\mathfrak{Pull}_{\widehat{\Xi}}}
\def\PS{\mathfrak{PS}_{\widehat{\Xi}}}
\def\Kerr{\mathfrak{Ker}_{\widehat{\Xi}}}

\def\nn{\nonumber}

\def\bbR{\mathbb{R}}
\def\bbC{\mathbb{C}}

\def\bbZ{\mathbb{Z}}

\def\EE{\mathcal{E}}

\def\NN{\mathcal{N}}

\def\Hom{\mathrm{Hom}}

\def\id{\mathrm{id}}

\def\dd{\mathrm{d}}
\def\vol{\mathrm{vol}}

\def\dim{\mathrm{dim}}

\def\1{\mathbbm{1}}
\def\oone{\mathbf{1}}
\def\g{\mathfrak{g}}

\def\ad{\mathrm{ad}}

\def\o{\mathfrak{o}}
\def\t{\mathfrak{t}}

\def\MW{\mathsf{MW}}

\newcommand{\ip}[2]{\left\langle #1,#2 \right\rangle}

\newcommand{\sect}[2]{\Gamma^\infty( #2 )}
\newcommand{\sectn}[2]{\Gamma_0^\infty( #2 )}

\title{%
A $C^\ast$-algebra for quantized principal $U(1)$-connections\\
on globally hyperbolic Lorentzian manifolds
}

\author{%
Marco Benini$^{1,2,a}$, Claudio Dappiaggi$^{1,b}$, Thomas-Paul Hack$^{3,c}$ 
and Alexander Schenkel$^{4,d}$\vspace{4mm}\\
{\small $^1$ Dipartimento di Fisica}\\ 
{\small Universit{\`a} di Pavia \& INFN, sezione di Pavia –- Via Bassi 6, 27100 Pavia, Italy.}\vspace{2mm}\\
{\small $^2$ II. Institut f\"ur Theoretische Physik}\\
{\small Universit\"at Hamburg,~Luruper Chausse 149,~22761~Hamburg,~Germany.}\vspace{2mm}\\
{\small $^3$ Dipartimento di Matematica}\\
{\small Universit{\`a} degli Studi di Genova –- Via Dodecaneso 35, 16146 Genova, Italy.}\vspace{2mm}\\
{\small $^4$ Fachgruppe Mathematik}\\
{\small Bergische~Universit\"at~Wuppertal,~Gau\ss stra\ss e~20,~42119~Wuppertal,~Germany.}\vspace{4mm}\\
 {\footnotesize  ~$^a$ marco.benini@pv.infn.it~,~$^b$ claudio.dappiaggi@unipv.it~,~$^c$ 
 hack@dima.unige.it~,~$^d$ schenkel@math.uni-wuppertal.de }
 }

\date{\today}

\begin{document}

\maketitle

\begin{abstract}
The aim of this work is to complete our program on the quantization of connections on arbitrary principal
$U(1)$-bundles over globally hyperbolic Lorentzian manifolds. In particular, we show that one can assign 
via a covariant functor to any such bundle an algebra of observables which separates gauge equivalence classes of connections. 
The $C^\ast$-algebra we construct generalizes the usual CCR-algebras since, contrary to
the standard field-theoretic models, it is based on a presymplectic Abelian group instead of a symplectic vector space. 
We prove a no-go theorem according to which neither this functor, nor any of its quotients, satisfies the strict axioms of 
general local covariance. As a byproduct, we prove that a morphism violates the locality axiom if and only if
a certain induced morphism of cohomology groups is non-injective.
We then show that fixing any principal $U(1)$-bundle, there exists a suitable category of sub-bundles 
for which a quotient of our functor yields a quantum field theory in the sense of Haag and Kastler.
We shall provide a physical interpretation of this feature and we obtain some new insights concerning electric charges
in locally covariant quantum field theory.
\end{abstract}


\section{\label{sec:intro}Introduction}
Although Maxwell's field is the simplest example of a Yang-Mills gauge theory,
it is known since \cite{Ashtekar:1980ki} that 
the construction and analysis of the associated algebra of observables and its representations
can be complicated due to a non-trivial topology of the spacetime manifold.
This peculiar feature is extremely relevant when one employs the algebraic framework in order to quantize 
such a theory on curved backgrounds. The first investigations along these lines are due to Dimock 
\cite{Dimock}, but a thorough analysis of topological effects started only 
recently, from both the perspective of the Faraday tensor \cite{DL} and, more generally, the
 quantization of linear gauge theories \cite{Pfenning:2009nx,Dappiaggi:2011cj,Fewster:2012bj,Hack:2012dm,SDH12, Finster:2013fva}. 
 The bottom line of some of these papers is the existence of a non-trivial center in the algebra of fields, provided 
 certain topological, or more precisely cohomological, properties of the underlying background hold true. 
 In \cite{SDH12}, it has been advocated that the elements of the center found in that paper 
 could be interpreted in physical terms as being related to observables measuring electric charges. However, this leads
  unavoidably to a violation of the locality property (injectivity of the induced morphisms between the 
  field algebras) of locally covariant quantum field theories, as formulated in \cite{Brunetti:2001dx}.

A complementary approach to the above ones has been introduced by some of us in \cite{Benini:2012vi,Benini:2013tra}
 starting from the observation that, in the spirit of a Yang-Mills gauge theory, electromagnetism 
 should be best described as a theory of connections on principal $U(1)$-bundles
  over globally hyperbolic Lorentzian manifolds. More properly, one starts from the characterization 
  of connections as sections of an affine bundle, dubbed the 
  {\em bundle of connections} \cite{Atiyah}. Subsequently the dynamics is implemented in terms of an 
  affine equation of motion, the Maxwell equation. The system can be quantized in the algebraic framework 
  following the prescription outlined in \cite{Benini:2012vi}. This procedure is advantageous 
  for three main reasons: First of all there is no need to fix any reference connection, 
  as it is done (implicitly) elsewhere \cite{Pfenning:2009nx,Dappiaggi:2011cj,SDH12,Finster:2013fva}. 
 As a useful consequence of this, we were able to construct in \cite{Benini:2013tra} purely topological observables,
 resembling topological quantum fields, which can measure the Chern class of the underlying principal $U(1)$-bundle.
  Secondly, interactions 
  between gauge and matter fields are modeled only in terms of connections, while an approach
   based on the Faraday tensor, as in \cite{DL}, cannot account for this aspect. Thirdly, contrary to 
   most of the previous approaches, the gauge group is completely determined geometrically by 
   the underlying principal bundle, since it is the collection of vertical automorphisms. 

By following this perspective, the algebra of fields for Abelian Yang-Mills theories has been constructed in \cite{Benini:2013tra}.
Yet, as explained in \cite[Remark 4.5]{Benini:2013tra}, the latter fails to separate gauge equivalence 
classes of connections. The source of this obstruction can be traced back to the existence of disconnected components
in the gauge group in the case of spacetimes with a non-trivial first de Rham cohomology group. 
From a physical point of view, this entails that those observables which are measuring the configurations tied
to the Aharonov-Bohm effect, as discussed in \cite{SDH12,Finster:2013fva}, are not contained in the algebra of observables. 

The main goal of this paper is to fill this gap by elaborating on the proposal in \cite{Benini:2013tra} 
to add Wilson-loop observables to the algebra of fields, as these new elements would solve the problem of separating 
 all  configurations. Following slavishly the original idea turned out to be rather cumbersome from a 
 technical point of view. Yet, we found that it is more convenient to consider exponentiated
  versions of the affine observables constructed
  in \cite{Benini:2013tra}. On the one hand, these observables resemble classical versions of Weyl operators, while, 
 on the other hand, the requirement of gauge invariance leads to a weaker 
 constraint -- the exponent does not need to remain invariant under a gauge transformation, 
 but it is allowed to change by any integer multiple of $2\pi i$. 

After performing this construction, we shall prove that, contrary to what was shown in
 \cite[Section 7]{Benini:2013tra} for the non-exponentiated algebra of fields, in the complete framework
 it is not possible to restore general local covariance in the strict sense by singling out a suitable ideal. 
This no-go theorem holds true only if we consider all possible isometric embeddings allowed by the axioms 
of general local covariance devised in \cite{Brunetti:2001dx}. If we restrict our
category of principal $U(1)$-bundles to a suitable
subcategory possessing a terminal object, a result similar to that of \cite[Section 7]{Benini:2013tra} can be shown to hold true.
  We will interpret this feature as a proof that we can construct a separating algebra of observables
  fulfilling the axioms of Haag and Kastler \cite{Haag:1963dh} generalized to an arbitrary but fixed globally hyperbolic spacetime.
  We shall further provide a physical interpretation for the impossibility to restore general local covariance in the 
 strict sense on our category of all principal $U(1)$-bundles.
 
We present an outline of the paper: In Section \ref{sec:prelim} we fix the notations and 
preliminaries which should allow a reader with some experience in differential 
geometry to follow the rest of the article. For more details, explanations and
 proofs we refer to \cite{Benini:2012vi,Benini:2013tra}, see also \cite{Baum,Kobayashi}
 for a general introduction to the differential geometry of gauge theories. 
 In Section \ref{sec:functionals} we provide 
 a detailed study of the exponential observables mentioned above. We characterize explicitly the
  gauge invariant exponential observables and prove that they separate gauge equivalence
   classes of connections. This solves the problem explained in \cite[Remark 4.5]{Benini:2013tra} 
   and captures the essence of what is called Aharonov-Bohm observables in  \cite{SDH12}.
As a rather unexpected result, we find that the set of gauge invariant exponential observables 
can be labeled by a presymplectic Abelian group, which is not a vector space due to the disconnected
components of the gauge group. We shall prove in Section \ref{sec:functor}
that these presymplectic Abelian groups naturally arise from a covariant functor from a category of principal $U(1)$-bundles
over globally hyperbolic spacetimes to a category of presymplectic Abelian groups.
 The properties of this functor are carefully investigated and it is found that, in agreement with earlier results
 \cite{SDH12,Benini:2013tra}, the locality property is violated (unless we restrict ourselves 
 to two-dimensional connected spacetimes). 
 As a novel and very important result, we shall give a precise characterization of which morphisms
 violate the locality property: Explicitly, we prove that a morphism violates the locality property
 if and only if a certain induced morphism between compactly supported 
 de Rham cohomology groups is not injective. This is a major improvement
 compared to earlier studies on the violation of the locality property.
After this we study whether our functor allows for a quotient by `electric charges' 
in order to overcome the failure of the locality property 
as it was done in \cite[Section 7]{Benini:2013tra}. We prove a no-go theorem: There exists no quotient
  such that the theory satisfies the locality property and we trace this feature back to Aharonov-Bohm observables,
  which were not present in \cite{Benini:2013tra}. In Section \ref{sec:quantization} we study the quantization
  of our presymplectic Abelian group functor in terms of the CCR-functor for presymplectic Abelian groups,
   which we develop in the Appendix \ref{sec:weyl} by applying and extending results of 
   \cite{Manuceau:1973yn}. The resulting quantum field theory functor 
  satisfies the quantum causality
  property and the quantum time-slice axiom, however not the locality property
   (again unless we restrict ourselves 
   to two-dimensional connected spacetimes). Our no-go theorem on the impossibility of
   curing the violation of the locality property by taking further quotients is extended to the quantum case.
   In Section \ref{sec:locality} 
  we consider suitable subcategories (possessing a terminal  object) of the category  of principal $U(1)$-bundles
   and prove that there exists a quotient which restores the locality property.
  The resulting theory is not a locally covariant quantum field theory in the strict definition of 
 \cite{Brunetti:2001dx}, but rather a theory in the sense of  Haag and Kastler \cite{Haag:1963dh}
where a global spacetime manifold (not necessarily the Minkowski spacetime) is fixed at the very beginning 
and one takes into account only causally compatible open sub-regions.
A physical interpretation of our results is given in Section \ref{sec:physical}.


\section{\label{sec:prelim}Preliminaries and notation}
Let us fix once and for all the Abelian Lie group $G=U(1)$. We denote its Lie algebra by $\g$ and notice that $\g = i\,\bbR$.
 The vector space dual of the Lie algebra $\g$ is denoted by $\g^\ast$ and we note that  $\g^\ast\simeq i\,\bbR$. 
 For later convenience we introduce the subgroup $\g_\bbZ:=2\pi\,i\,\bbZ\subset \g$ (of the Abelian group $(\g,+)$),
 which is a lattice in $\g$.

In \cite[Definition 2.4]{Benini:2013tra} we have defined a suitable category $G{-}\PrBu$ of principal $G$-bundles
over globally hyperbolic spacetimes, which provides a natural arena to study field theories of principal $G$-connections.
An object in $G{-}\PrBu$ is a tuple $\Xi= ((M,\o,g,\t),(P,r))$, 
where $(M,\o,g,\t)$ is an oriented and time-oriented globally hyperbolic Lorentzian manifold\footnote{
We further assume that $\dim(M)\geq 2$ and that $M$ is of finite type, which means that
$M$ has a finite good cover, i.e.\ an open cover by contractible subsets such that all (multiple) overlaps are also contractible.
} and $(P,r)$ is a principal $G$-bundle over $M$. A morphism $f:\Xi_1\to \Xi_2$ in $G{-}\PrBu$ 
is a principal $G$-bundle map $f:P_1\to P_2$,
such that the induced map $\underline{f}:M_1\to M_2$ is an orientation and time-orientation preserving isometric
 embedding with $\underline{f}[M_1]\subseteq M_2$  causally compatible and open. Remember that a subset
 $S\subseteq M$ is called causally compatible if $J_S^\pm(\{x\}) = J_M^\pm(\{x\})\cap S$, for all $x\in S$,
 where $J^\pm_S$ and $J^\pm_M$ denotes the causal future/past in $S$ and $M$, respectively 
 (notice that $S$ is an oriented and time-oriented Lorentzian manifold by pulling back these data from $M$).
We shall later also require the following full subcategories of $G{-}\PrBu$: 
We denote by $G{-}\PrBu^{(m)}$, with $m\geq 2$, the full subcategory of $G{-}\PrBu$,
such that for each object $\Xi$ the underlying spacetime has dimension $\dim(M)=m$.
Furthermore, we denote by $G{-}\PrBu^{(m)}_0$ the full subcategory of $G{-}\PrBu$, 
such that for each object $\Xi$ the underlying spacetime $M$ is connected and has dimension $\dim(M)=m$.

 To any object $\Xi$ in $G{-}\PrBu$ we can associate (via a covariant functor)
 its bundle of connections $\mathcal{C}(\Xi)$, that is 
an affine bundle over $M$ modeled on the homomorphism
 bundle $\Hom(TM,\ad(\Xi))$. Notice that the adjoint bundle is trivial, i.e.~$\ad(\Xi) = M\times \g$, since $G$ is Abelian.
 The set of sections $\sect{M}{\mathcal{C}(\Xi)}$ of the bundle $\mathcal{C}(\Xi)$ is an (infinite-dimensional)
 affine space over the vector space of $\g$-valued one-forms $\Omega^1(M,\g)$. We denote the free and transitive action
 of $\Omega^1(M,\g)$ on $\sect{M}{\mathcal{C}(\Xi)}$ with the usual abuse of notation by  $\lambda + \eta$, for all
 $\eta\in \Omega^1(M,\g)$ and $\lambda \in \sect{M}{\mathcal{C}(\Xi)}$.
 Let us denote by $\sectn{M}{\mathcal{C}(\Xi)^\dagger}$ the vector space of compactly supported sections
 of the vector dual bundle $\mathcal{C}(\Xi)^\dagger$, which is the vector bundle of affine homomorphisms
 from $\mathcal{C}(\Xi)$ to $M\times\bbR$ (i.e.\ the fibre at $x\in M$ of $\mathcal{C}(\Xi)^\dagger$ is the 
 vector space of affine maps $\mathcal{C}(\Xi)\vert_x \to \bbR$). 
 As a consequence of the fibre-wise duality pairing between $\mathcal{C}(\Xi)^\dagger$ and $\mathcal{C}(\Xi)$, 
 every compactly supported section $\varphi\in\sectn{M}{\mathcal{C}(\Xi)^\dagger}$ defines a functional
 on the configuration space $\sect{M}{\mathcal{C}(\Xi)}$ by
 \begin{flalign}\label{eqn:affobs}
 \mathcal{O}_\varphi : \sect{M}{\mathcal{C}(\Xi)} \to \bbR~,~~\lambda \mapsto \mathcal{O}_\varphi(\lambda) = \int_M \varphi(\lambda)\,\vol~.
 \end{flalign}
 For any $\lambda\in \sect{M}{\mathcal{C}(\Xi)}$ and $\eta\in \Omega^1(M,\g)$ the functional 
 $\mathcal{O}_\varphi$ satisfies the affine property $\mathcal{O}_\varphi(\lambda+\eta) 
 = \mathcal{O}_\varphi(\lambda) + \ip{\varphi_V}{\eta}$,
where 
\begin{flalign}
\ip{\varphi_V}{\eta} := \int_M \varphi_V \wedge \ast(\eta)~.
\end{flalign}
We have denoted the Hodge operator by $\ast$ and the linear part of $\varphi\in \sectn{M}{\mathcal{C}(\Xi)^\dagger}$ 
by $\varphi_V\in \Omega^1_0(M,\g^\ast)$.\footnote{
By linear part of a section $\varphi\in \sectn{M}{\mathcal{C}(\Xi)^\dagger}$ we always mean the section 
$\varphi_V\in \Omega^1_0(M,\g^\ast)$ that is canonically obtained by taking point-wise the linear part of the affine map
$\varphi(x): \mathcal{C}(\Xi)\vert_x \to \bbR$, which is a linear map $\varphi(x)_V: T^\ast M\vert_x \times \g \to \bbR  $
that can be identified (by using the metric $g$) with an element in $T^\ast M\vert_x \times \g^\ast$.
}
The duality pairing between $\g^\ast$ and $\g$ is suppressed here and in the following.
Let us define the vector subspace
\begin{flalign}
\mathrm{Triv}:= \Big\{a\,\1: a\in C^\infty_0(M) \text{~satisfies~}\int_M a\,\vol =0\Big\} \subseteq \sectn{M}{\mathcal{C}(\Xi)^\dagger}~,
\end{flalign}
where $\1\in \sect{M}{\mathcal{C}(\Xi)^\dagger}$ denotes the canonical section which associates to any $x\in M$ the
constant affine map $\mathcal{C}(\Xi) \vert_x \ni \lambda \mapsto 1$.
Notice that any $\varphi \in \mathrm{Triv}$ defines the trivial functional $\mathcal{O}_\varphi \equiv 0$
and, vice versa, that for any trivial functional $\mathcal{O}_\varphi\equiv 0$ we have $\varphi\in \mathrm{Triv}$.
Hence, the quotient $\sectn{M}{\mathcal{C}(\Xi)^\dagger}/\mathrm{Triv}$ labels
distinct affine functionals (\ref{eqn:affobs}).
Elements in this quotient are equivalence classes that we denote by $\varphi$ (suppressing square brackets)
in order to simplify notation.

The gauge group $\mathrm{Gau}(P)$, i.e.~the group of vertical principal $G$-bundle automorphisms,
is isomorphic to the group $C^\infty(M,G)$, 
which acts on $\sect{M}{\mathcal{C}(\Xi)}$ via
\begin{flalign}\label{eqn:gaugetrafo}
\sect{M}{\mathcal{C}(\Xi)} \times C^\infty(M,G) \to \sect{M}{\mathcal{C}(\Xi)} ~,~~(\lambda,\widehat{f}) \mapsto 
\lambda + \widehat{f}^{\ast}(\mu_G)~,
\end{flalign}
where $\mu_G \in \Omega^1(G,\g)$ is the Maurer-Cartan form and $\widehat{f}^\ast : \Omega^1(G,\g)\to \Omega^1(M,\g)$
 denotes the pull-back.
We call the transformations in (\ref{eqn:gaugetrafo}) {\it gauge transformations}.
Let us define the subgroup of the Abelian group $(\Omega^1(M,\g),+)$ which is generated by gauge transformations,
\begin{flalign}\label{eqn:gaugespace00}
B_G := \big\{\widehat{f}^\ast(\mu_G) : \widehat{f}\in C^\infty(M,G)\big\}~.
\end{flalign}
Notice that since the Maurer-Cartan form is closed, i.e.\ $\dd \mu_G=0$, we have
$B_G\subseteq \Omega^1_\dd(M,\g)$, where $\Omega^1_\dd(M,\g)$ denotes the vector space of $\g$-valued
closed one-forms. Furthermore, since any $\chi\in C^\infty(M,\g)$ can be exponentiated to an element in the gauge group
$\exp\circ\chi \in C^\infty(M,G)$, the gauge transformations (\ref{eqn:gaugetrafo}) in particular include
all transformations of the form $\lambda \mapsto \lambda + \dd\chi$, with $\chi\in C^\infty(M,\g)$.
Hence, $\dd C^\infty(M,\g) \subseteq B_G\subseteq \Omega^1_\dd(M,\g)$.
In order to give a precise characterization of the Abelian group $B_G$  we are going to use \v{C}ech cohomology
(see also \cite[Proposition 4.2]{Benini:2013tra} for a more abstract argument leading to the same results):
Let  $\mathcal{U} := \{U_\alpha\}_{\alpha\in \mathcal{I}}$ 
be any good open cover of $M$ and let us denote by $\check{H}^1(\mathcal{U},\g_\bbZ)$ 
the first \v{C}ech cohomology group of $\mathcal{U}$ with values in the constant presheaf $\g_\bbZ = 2\pi i\,\bbZ$. 
Notice that the Abelian group $\check{H}^1(\mathcal{U},\g_\bbZ)$ is a free $\bbZ$-module,
which is finitely generated because $M$ is assumed to be of finite type. 
Due to the canonical embedding $\g_\bbZ \hookrightarrow \g$ 
there exists a monomorphism of Abelian groups $\check{H}^1(\mathcal{U},\g_\bbZ) \to \check{H}^1(\mathcal{U},\g)$ 
into the first \v{C}ech cohomology group of $\mathcal{U}$ with values in the constant presheaf $\g = i\,\bbR$.
The latter is isomorphic to the first de Rham cohomology group $H^1_{\mathrm{dR}}(M,\g)$
via the \v{C}ech-de Rham isomorphism, which is given by
the following construction: For $[\eta]\in H^1_{\mathrm{dR}}(M,\g)$ choose any representative $\eta \in \Omega^1_\dd(M,\g)$.
Restricting $\eta$ to the $U_{\alpha}$, there exist $\chi_\alpha\in C^\infty(U_\alpha,\g)$, such that
$\eta\vert_{U_\alpha} = \dd \chi_\alpha$. On the overlaps $U_\alpha\cap U_\beta$
the difference $\eta_{\alpha\beta} =\chi_\alpha -\chi_\beta \in \g$ is constant, hence it defines
a \v{C}ech $1$-cocycle $\{\eta_{\alpha\beta}\}$ and therewith an element
 $[\{\eta_{\alpha\beta}\}] \in \check{H}^1(\mathcal{U},\g)$. The inverse of the \v{C}ech-de Rham isomorphism 
 is given by using a partition of unity $\{\psi_\alpha\}$ subordinated to $\mathcal{U}$ 
and sending $[\{\eta_{\alpha\beta}\}] \in \check{H}^1(\mathcal{U},\g)$ to the de Rham class 
$[\eta]\in H^1_{\mathrm{dR}}(M,\g)$, where the differential form 
$\eta\in \Omega^1_\dd(M,\g)$ is defined by setting 
$\eta\vert_{U_\alpha} = \dd \big(\sum_{\beta\in\mathcal{I}} \eta_{\alpha\beta}\psi_\beta\big)$.
We denote the image of the subgroup $\check{H}^1(\mathcal{U},\g_\bbZ)\subset \check{H}^1(\mathcal{U},\g)$
under the \v{C}ech-de Rham isomorphism by $H_\mathrm{dR}^1(M,\g_\bbZ)$
and notice that it is a lattice in $H^1_{\mathrm{dR}}(M,\g)$,
i.e.~any $\bbZ$-module basis of $H^1_{\mathrm{dR}}(M,\g_\bbZ)$ provides a vector space basis of $H^1_{\mathrm{dR}}(M,\g)$.
Using the \v{C}ech-de Rham isomorphism we observe that, 
for all $\widehat{f}\in C^\infty(M,G)$, $[\widehat{f}^\ast(\mu_G)]\in H^1_{\mathrm{dR}}(M,\g_\bbZ)$:
Indeed, restricting $\widehat{f}$ to the $U_\alpha$, there exist $\chi_\alpha\in C^\infty(U_\alpha,\g)$ such that
$\widehat{f}\vert_{U_\alpha} = \exp \circ \chi_\alpha$ and hence $\widehat{f}^\ast(\mu_G)\vert_{U_\alpha} =\dd\chi_\alpha$
with $\eta_{\alpha\beta} = \chi_\alpha - \chi_\beta \in \g_\bbZ =2\pi i\, \bbZ$ on $U_\alpha\cap U_\beta$.
On the other hand, any element in $H^1_{\mathrm{dR}}(M,\g_\bbZ)$ has a representative of the form $\widehat{f}^\ast(\mu_G)$ with
$\widehat{f}\in C^\infty(M,G)$: Indeed, for any integral \v{C}ech $1$-cocycle $\{\eta_{\alpha\beta}\in \g_\bbZ\}$
we can construct $\widehat{f}\in C^\infty(M,G)$ by setting $\widehat{f}\vert_{U_\alpha} = \exp \circ \chi_\alpha$
with $\chi_\alpha = \sum_{\beta\in\mathcal{I}} \eta_{\alpha\beta}\psi_\beta$ for some choice of partition of
unity $\{\psi_\alpha\}$ subordinated to $\mathcal{U}$. Via the \v{C}ech-de Rham isomorphism,
the de Rham class $[\widehat{f}^\ast(\mu_G)]$ is identified with the \v{C}ech class $[\{\eta_{\alpha\beta}\}]$.
Hence, we have obtained the explicit characterization  
\begin{flalign}\label{eqn:gaugespace}
B_G = \big\{ \eta\in \Omega^1_\dd(M,\g) : [\eta] \in H_\mathrm{dR}^1(M,\g_\bbZ)\big\}~.
\end{flalign}

 The gauge invariant affine functionals (\ref{eqn:affobs}) have been characterized in \cite[Theorem 4.4]{Benini:2013tra}.
 It is found that these functionals are labeled by those $\varphi \in \sectn{M}{\mathcal{C}(\Xi)^\dagger}/\mathrm{Triv}$ which satisfy
 $\varphi_V\in\delta \Omega^2_0(M,\g^\ast)$, where $\delta$ is the codifferential. 
 As a consequence of \cite[Remark 4.5]{Benini:2013tra}, these functionals in general do not separate gauge equivalence
 classes of connections.
 The goal of the present article is to resolve this issue by studying a set of  observables
 different from (\ref{eqn:affobs}).


\section{\label{sec:functionals}Gauge invariant exponential functionals}
Instead of  (\ref{eqn:affobs}), let us consider the exponential functionals,
for all $\varphi \in \sectn{M}{\mathcal{C}(\Xi)^\dagger}$,
\begin{flalign}\label{eqn:expobs}
\mathcal{W}_\varphi : \sect{M}{\mathcal{C}(\Xi)} \to \bbC~,~~\lambda \mapsto \mathcal{W}_\varphi(\lambda) = 
e^{2\pi i\,\mathcal{O}_\varphi(\lambda)}~.
\end{flalign}
The affine property of $\mathcal{O}_\varphi$ implies that, 
for all $\lambda\in\sect{M}{\mathcal{C}(\Xi)}$ and $\eta\in\Omega^1(M,\g)$,
\begin{flalign}\label{eqn:expobsaff}
\mathcal{W}_\varphi(\lambda + \eta ) = \mathcal{W}_\varphi(\lambda) \, e^{2\pi i\,\ip{\varphi_V}{\eta}}~.
\end{flalign}
We notice that the functional $\mathcal{W}_\varphi$ is trivial, i.e.~$\mathcal{W}_\varphi \equiv 1$, if and only if
$\varphi$ is an element in the subgroup
\begin{flalign}
\mathrm{Triv}_\bbZ := \Big\{a\,\1 : a\in C^\infty_0(M) \text{~satisfies~}\int_M a\,\vol \in \bbZ \Big\}
 \subseteq \sectn{M}{\mathcal{C}(\Xi)^\dagger}~.
\end{flalign}
Hence, we consider the quotient
\begin{flalign}
\EE^\mathrm{kin}:= \sectn{M}{\mathcal{C}(\Xi)^\dagger}/\mathrm{Triv}_\bbZ~
\end{flalign}
in order to label distinct exponential functionals.
Elements in this quotient are equivalence classes that we simply denote by $\varphi$ (suppressing square brackets).

We say that a functional $\mathcal{W}_\varphi$, $\varphi\in \EE^\mathrm{kin}$, is {\it gauge invariant}, if 
$\mathcal{W}_\varphi(\lambda+\eta) = \mathcal{W}_\varphi(\lambda)$, for 
all $\lambda \in \sect{M}{\mathcal{C}(\Xi)}$ and $\eta\in B_G$.
Due to (\ref{eqn:expobsaff}) this is equivalent to $\ip{\varphi_V}{B_G}\subseteq \bbZ$.
A necessary condition for $\mathcal{W}_\varphi$ to be gauge invariant is that $\delta\varphi_V =0$, 
i.e.~$\varphi_V\in\Omega^1_{0,\delta}(M,\g^\ast)$, where by the subscript $_{\delta}$ we denote co-closed forms. 
This can be seen by demanding invariance of $\mathcal{W}_\varphi$ under the gauge 
transformations  $\lambda \mapsto \lambda +\dd \chi$, $\chi\in C^\infty(M,\g)$, 
which are obtained by choosing $\widehat{f} = \exp\circ \chi \in C^\infty(M,G)$ in (\ref{eqn:gaugetrafo}).
We can associate to such $\varphi_V$ an element $[\varphi_V]$
in the dual de Rham cohomology group $H^1_{0\,\mathrm{dR}^\ast}(M,\g^\ast) 
:= \Omega^1_{0,\delta}(M,\g^\ast)/ \delta \Omega^2_0(M,\g^\ast)$.
Since any $\eta \in  B_G$ is closed, the pairing in (\ref{eqn:expobsaff}) depends only on the cohomology classes,
i.e.~$\ip{\varphi_V}{\eta} = \ip{[\varphi_V]}{[\eta]}$. Notice that the pairing
$\ip{~}{~}: H^1_{0\,\mathrm{dR}^\ast}(M,\g^\ast) \times H^1_{\mathrm{dR}}(M,\g) \to \bbR$ is non-degenerate
due to Poincar{\'e} duality, i.e.~$H^1_{0\,\mathrm{dR}^\ast}(M,\g^\ast) \simeq H^1_{\mathrm{dR}}(M,\g)^\ast :=
\Hom_\bbR(H^1_{\mathrm{dR}}(M,\g),\bbR)$.

Since the gauge transformations are characterized by an integral cohomology condition (\ref{eqn:gaugespace}),
also the gauge invariant exponential functionals will be characterized by some integral cohomology condition.
Before we can determine the exact form of this condition, we need some notations:
Let us denote the dual $\bbZ$-module of $H^1_\mathrm{dR}(M,\g_\bbZ)$ by $H^1_\mathrm{dR}(M,\g_\bbZ)^\ast := 
\Hom_\bbZ(H^1_\mathrm{dR}(M,\g_\bbZ),\bbZ)$. Since $H^1_\mathrm{dR}(M,\g_\bbZ)$ is a lattice
in $H^1_\mathrm{dR}(M,\g)$ any element in $H^1_\mathrm{dR}(M,\g_\bbZ)^\ast$
defines a unique element in $H^1_\mathrm{dR}(M,\g)^\ast$ by $\bbR$-linear extension.
Thus, there is a monomorphism of Abelian groups $H^1_\mathrm{dR}(M,\g_\bbZ)^\ast \to H^1_\mathrm{dR}(M,\g)^\ast$
which we shall suppress in the following. Composing this map with the isomorphism $H^1_{0\,\mathrm{dR}^\ast}(M,\g^\ast) 
\simeq H^1_{\mathrm{dR}}(M,\g)^\ast$ given by the pairing $\ip{~}{~}$ we can regard
$H^1_\mathrm{dR}(M,\g_\bbZ)^\ast $ as a subgroup of $H^1_{0\,\mathrm{dR}^\ast}(M,\g^\ast)$, which
we shall denote by $H^1_{0\,\mathrm{dR}^\ast}(M,\g^\ast)_\bbZ \subseteq H^1_{0\,\mathrm{dR}^\ast}(M,\g^\ast)$.
With these preparations we can now provide an explicit characterization of the gauge invariant exponential functionals.
\begin{propo}\label{propo:gaugeinv}
Let $\varphi \in \EE^\mathrm{kin}$ be such that $\delta\varphi_V=0$, i.e.~$\varphi$ satisfies the necessary condition
for $\mathcal{W}_\varphi$ being gauge invariant. Then $\mathcal{W}_\varphi$ is a gauge invariant
functional if and only if $[\varphi_V] \in H^1_{0\,\mathrm{dR}^\ast}(M,\g^\ast)_\bbZ$.
\end{propo} 
\begin{proof}
The functional (\ref{eqn:expobs}) is gauge invariant if and only if $\ip{\varphi_V}{B_G}=\ip{[\varphi_V]}{[B_G]}
\subseteq \bbZ$.
By (\ref{eqn:gaugespace}) this is equivalent to
the condition 
$\ip{[\varphi_V]}{H^1_\mathrm{dR}(M,\g_\bbZ)}\subseteq \bbZ$, which is satisfied if and only if
 $[\varphi_V]\in H^1_{0\,\mathrm{dR}^\ast}(M,\g^\ast)_\bbZ $.
\end{proof}
Let us define the subgroup
\begin{flalign}\label{eqn:Einv}
\EE^\mathrm{inv}:= \big\{\varphi \in \EE^\mathrm{kin} : \delta\varphi_V=0\text{~and~} [\varphi_V]\in H^1_{0\,\mathrm{dR}^\ast}(M,\g^\ast)_\bbZ\big\}\subseteq \EE^\mathrm{kin}~,
\end{flalign}
which labels the gauge invariant functionals $\mathcal{W}_\varphi$.
\begin{theo}\label{thmSep}
The set $\{\mathcal{W}_\varphi: \varphi \in \EE^\mathrm{inv}\}$ of gauge invariant exponential functionals is
separating on gauge equivalence classes of configurations. This means that, for any two 
$\lambda,\lambda^\prime \in \sect{M}{\mathcal{C}(\Xi)}$ which are not gauge equivalent via (\ref{eqn:gaugetrafo}), there
exists $\varphi\in \EE^\mathrm{inv}$, such that $\mathcal{W}_\varphi(\lambda^\prime)\neq \mathcal{W}_\varphi(\lambda)$.
\end{theo}
\begin{proof}
Let $\lambda,\lambda^\prime \in \sect{M}{\mathcal{C}(\Xi)}$ be not gauge equivalent, 
i.e.~$\lambda^\prime = \lambda + \eta$ with $\eta\in \Omega^1(M,\g)\setminus B_G$.

Let us first assume that $\eta$ is not closed, $\dd\eta\neq 0$. For all $\zeta \in \Omega^2_0(M,\g^\ast)$ let us consider 
$\underline{\mathcal{F}}^\ast(\zeta)\in \EE^\mathrm{kin}$, where $\underline{\mathcal{F}}^\ast:
\Omega^2_0(M,\g^\ast) \to \EE^\mathrm{kin}$
is the formal adjoint of the curvature affine differential operator
$\underline{\mathcal{F}}: \sect{M}{\mathcal{C}(\Xi)} \to \Omega^2(M,\g)$
 (cf.~\cite[Lemma 2.14 and Proposition 2.18]{Benini:2013tra}). Notice that
 $\underline{\mathcal{F}}^\ast(\zeta)_V = -\delta \zeta $ (for an explanation for the
 minus sign see \cite[Proposition 2.18]{Benini:2013tra} and the subsequent discussion), 
 hence $\underline{\mathcal{F}}^\ast(\zeta)\in \EE^\mathrm{inv}$.
 We obtain for the corresponding functional
 \begin{flalign}
 \mathcal{W}_{\underline{\mathcal{F}}^\ast(\zeta)}(\lambda^\prime) = 
 \mathcal{W}_{\underline{\mathcal{F}}^\ast(\zeta)}(\lambda)\,e^{-2\pi i\,\ip{\zeta}{\dd\eta}}~.
 \end{flalign}
 Since $\dd\eta\neq 0$ there exists $\zeta \in \Omega^2_0(M,\g^\ast)$ such that
 $\mathcal{W}_{\underline{\mathcal{F}}^\ast(\zeta)}(\lambda^\prime) \neq 
 \mathcal{W}_{\underline{\mathcal{F}}^\ast(\zeta)}(\lambda)$.

Let us now assume that $\dd\eta =0$.  
By hypothesis, the corresponding cohomology class
$[\eta]\in H^1_\mathrm{dR}(M,\g)$ is not included in the subgroup
$H^1_\mathrm{dR}(M,\g_\bbZ)\subseteq H^1_\mathrm{dR}(M,\g)$, since otherwise $\eta$
would be an element in $B_G$. We prove the statement by contradiction:
Assume that $\mathcal{W}_\varphi(\lambda^\prime) = \mathcal{W}_\varphi(\lambda)$, for all
$\varphi\in \EE^\mathrm{inv}$. As a consequence, 
$\ip{H^1_{0\,\mathrm{dR}^\ast}(M,\g^\ast)_\bbZ}{[\eta]}\subseteq \mathbb{Z}$,
which implies that $[\eta]$ defines a homomorphism of Abelian groups
$H^1_{0\,\mathrm{dR}^\ast}(M,\g^\ast)_\bbZ \to \bbZ$. Notice that this is an element in the double
dual $\bbZ$-module of $H^1_\mathrm{dR}(M,\g_\bbZ)$, which 
is isomorphic to $H^1_\mathrm{dR}(M,\g_\bbZ)$ since the latter is 
 finitely generated and free. This is a contradiction and hence
there exists $\varphi\in \EE^\mathrm{inv}$, such that
 $\mathcal{W}_\varphi(\lambda^\prime) \neq \mathcal{W}_\varphi(\lambda)$.
\end{proof}

\begin{rem} \label{rem:wilson}
There is the following relation to the usual Wilson loop observables: 
Given a smooth loop $\gamma: \mathbb{S}^1 \to M$
we can construct the pull-back bundle $\gamma^\ast(P)$, which is a principal $U(1)$-bundle over
$\mathbb{S}^1$. By construction, we have the commuting diagram
\begin{flalign}
\xymatrix{
\gamma^\ast(P)\ar[d]_-{\pi^\prime} \ar[rr]^-{\overline{\gamma}} && P \ar[d]^-{\pi}\\
\mathbb{S}^1 \ar[rr]^\gamma && M
}
\end{flalign}
Notice that $\gamma^\ast(P)$ is necessarily a trivial bundle (as $H^2(\mathbb{S}^1,\bbZ)=\{0\}$) and hence there
exists a global section $\sigma: \mathbb{S}^1 \to \gamma^\ast(P)$ of $\pi^\prime$. 
Given any $\lambda\in \sect{M}{\mathcal{C}(\Xi)}$,  its associated connection form $\omega_\lambda\in\Omega^1(P,\g)$ 
pulls back to a connection form $\overline{\gamma}^\ast(\omega_\lambda)\in\Omega^1(\gamma^\ast(P),\g)$,
which can be further pulled back via the section to a $\g$-valued one-form on $\mathbb{S}^1$,
$\sigma^\ast(\overline{\gamma}^\ast(\omega_\lambda))\in \Omega^1(\mathbb{S}^1,\g)$. We call the functional
\begin{flalign}\label{eqn:Wilsonloop}
w_\gamma: \sect{M}{\mathcal{C}(\Xi)} \to \bbC~,~~\lambda \mapsto w_\gamma(\lambda)=
e^{\int_{\mathbb{S}^1} \sigma^\ast(\overline{\gamma}^\ast(\omega_\lambda))}
\end{flalign}
a Wilson loop observable and notice that $w_\gamma$ does not depend on the choice of trivialization $\sigma$.
The exponent of the Wilson loop observables is an affine functional,  
for all $\lambda\in \sect{M}{\mathcal{C}(\Xi)}$ and $\eta \in \Omega^1(M,\g)$, 
$\int_{\mathbb{S}^1} \sigma^\ast(\overline{\gamma}^\ast(\omega_{\lambda+\eta})) = \int_{\mathbb{S}^1}
 \sigma^\ast(\overline{\gamma}^\ast(\omega_\lambda)) + \int_{\mathbb{S}^1} \gamma^\ast(\eta)$.
This immediately implies that
\begin{flalign}
w_\gamma(\lambda +\eta) = w_\gamma(\lambda) \,e^{\int_{\mathbb{S}^1} \gamma^\ast(\eta)}~,
\end{flalign}
which the reader should compare with (\ref{eqn:expobsaff}). Hence, the usual Wilson loop observables
(\ref{eqn:Wilsonloop})  can be regarded as exponential functionals (\ref{eqn:expobs}) obtained by
using distributional sections of the vector dual bundle $\mathcal{C}(\Xi)^\dagger$.
In our work we shall discard these distributional functionals and only work with smooth sections 
of the vector dual bundle $\mathcal{C}(\Xi)^\dagger$ for the following reasons:
Firstly, because of Theorem \ref{thmSep} the set of gauge invariant 
observables $\{\mathcal{W}_\varphi : \varphi\in \EE^\mathrm{inv}\}$
is already large enough to separate gauge equivalence classes of connections, hence we see no reason
to extend it by allowing for distributional sections of $\mathcal{C}(\Xi)^\dagger$.
Secondly, allowing for distributional sections will lead to singularities in our quantization prescription, 
the renormalization of which we would like to avoid in this paper. 
\end{rem}

For $k=0,\dots, \dim(M)$, let $\square_{(k)} := \delta\circ \dd + \dd\circ \delta : \Omega^k(M,\g^\ast)\to \Omega^k(M,\g^\ast)$
be the Hodge-d'Alembert operator acting on $\g^\ast$-valued $k$-forms.
As these operators are normally hyperbolic, they have unique retarded and advanced Green's operators denoted by
$G_{(k)}^\pm : \Omega^k_0(M,\g^\ast)\to \Omega^k(M,\g^\ast)$, see \cite{Bar:2007zz,Pfenning:2009nx} for details.
It is easy to prove that the d'Alembert operators $\square_{(k)}$ and the Green's operators 
 $G_{(k)}^\pm$ commute with the differential and codifferential, i.e.\
\begin{subequations}\label{eqn:commuteddeltaG}
\begin{flalign}
\dd \circ \square_{(k)} = \square_{(k+1)}\circ \dd \quad &,\qquad \delta\circ \square_{(k+1)}= \square_{(k)}\circ \delta~,\\
\dd \circ G^\pm_{(k)} = G^\pm_{(k+1)} \circ \dd \quad&,\qquad \delta\circ G^\pm_{(k+1)}= G^\pm_{(k)}\circ \delta~.
\end{flalign}
\end{subequations}
We denote the causal propagator by
$G_{(k)}:= G_{(k)}^+ - G_{(k)}^-: \Omega^k_0(M,\g^\ast)\to \Omega^k(M,\g^\ast)$
and notice that also the $G_{(k)}$ commute with $\dd$ and $\delta$ as a consequence of (\ref{eqn:commuteddeltaG}).
Given further a bi-invariant Riemannian metric $h$ on the structure group $G$ (or equivalently a $G$-equivariant positive 
linear map $h:\g\to\g^\ast$), 
we can define a presymplectic structure $\tau:\EE^\mathrm{inv} \times \EE^\mathrm{inv}\to \bbR$ 
on the Abelian group $\EE^\mathrm{inv}$ by,
for all $\varphi,\psi\in \EE^\mathrm{inv}$,
\begin{flalign}\label{eqn:taumap}
\tau(\varphi,\psi) := \ip{\varphi_V}{G_{(1)}(\psi_V)}_h^{} :=
 \int_M \varphi_V \wedge \ast\big(h^{-1}\big(G_{(1)}(\psi_V)\big)\big) ~,
\end{flalign}
where $h^{-1}: \g^\ast \to \g$ is the inverse of $h$.
This presymplectic structure can be derived from the Lagrangian density $\mathcal{L}[\lambda] = 
-\frac{1}{2} \,h(\underline{\mathcal{F}}(\lambda))\wedge \ast(\underline{\mathcal{F}}(\lambda))$
by slightly adapting Peierls' method \cite[Remark 3.5]{Benini:2013tra}.
It is worth mentioning that since $\g = i\,\bbR$ is one-dimensional and the adjoint action of $G$ on $\g$ is trivial, 
 bi-invariant Riemannian metrics on $G$ are in bijective correspondence with positive linear maps 
 $h:\g \to \g^\ast\,,~t\mapsto h(t) = \frac{1}{q^2}\,t $, where $q\in (0,\infty)$. 
Hence, the metric $h$ plays the role of an electric charge constant. To see this, plug $h(t) = \frac{1}{q^2}\,t$ into the Lagrangian above
and compare it with the usual textbook Lagrangian of Maxwell's theory. The metric $h$ will be fixed throughout this work.

Before we take  the quotient of $\EE^\mathrm{inv}$ by a  subgroup containing
the equation of motion, let us study the elements $\psi\in \EE^\mathrm{inv}$
which lead to central Weyl symbols in the quantum field theory. The Weyl relations
(\ref{eqn:Weylrelations}) read 
$W(\varphi) \,W(\psi) = e^{-i\,\tau(\varphi,\psi)/2} \,W(\varphi+\psi) $.
$W(\psi)$ commutes with all other Weyl symbols if and only if $\tau(\EE^\mathrm{inv},\psi)\subseteq 2\pi \,\mathbb{Z}$.
We denote by $\NN\subseteq \EE^\mathrm{inv}$ the subgroup of all $\psi\in \EE^\mathrm{inv}$
satisfying this condition, i.e.\
\begin{flalign}\label{eqn:N}
\NN :=\big\{ \psi\in \EE^\mathrm{inv} : \tau(\EE^\mathrm{inv},\psi)\subseteq 2\pi \,\mathbb{Z}\big\} ~.
\end{flalign}
This subgroup can also be characterized as follows:
\begin{propo}\label{propo:null1}
$\NN = \big\{\psi\in \EE^\mathrm{inv}: \psi_V \in \delta\Omega^2_{0,\dd}(M,\g^\ast) ~\text{and}~
\big[h^{-1}\big(G_{(1)}(\psi_V)\big)\big]
\in 2\pi\, H^1_\mathrm{dR}(M,\g_\bbZ) \big\}$.
\end{propo}
\begin{proof}
We first prove the inclusion $\supseteq$.
Assume that $\psi\in\EE^\mathrm{inv}$ satisfies the first condition of the Abelian group specified on the right hand side above,
i.e.~$\psi_V = \delta\zeta$ for some $\zeta\in \Omega^2_{0,\dd}(M,\g^\ast)$.
Then $\dd( h^{-1}(G_{(1)}(\psi_V)))=h^{-1}(G_{(2)}(\dd\delta\zeta)) =
h^{-1}(G_{(2)}(\square_{(2)}(\zeta)))= 0$, thus the second condition is well-posed.
Using Proposition \ref{propo:gaugeinv} the following holds true,
\begin{flalign}\label{eqn:tmpNull}
 \tau(\EE^\mathrm{inv},\psi)=\ip{H^1_{0\,\mathrm{dR}^\ast}(M,\g^\ast)_\bbZ}{\big[h^{-1}\big(G_{(1)}(\psi_V)\big)\big]}
 \subseteq 2\pi\,\bbZ~.
\end{flalign}
To prove the inclusion $\subseteq$, 
suppose that $\psi\in\EE^\mathrm{inv}$ is such that 
$\tau(\EE^\mathrm{inv},\psi)\subseteq 2\pi\,\bbZ$. Since $\EE^\mathrm{inv}$
contains the subgroup 
$\{\varphi\in \EE^\mathrm{kin}: \varphi_V \in \delta \Omega^2_0(M,\g^\ast)\}$,
 we obtain that $\dd( h^{-1}(G_{(1)}(\psi_V)) )= 0 $.
As a consequence of global hyperbolicity and $\square_{(2)}$ being normally hyperbolic,
we obtain $\dd\psi_V = \square_{(2)}(\zeta)$ for some $\zeta\in \Omega^2_0(M,\g^\ast)$.
If $\dim(M)\geq 3$, we apply $\dd$ to this equation and find $0=\square_{(3)}(\dd\zeta)$,  which again due to
global hyperbolicity and $\square_{(3)}$ being normally hyperbolic
implies that $\zeta\in \Omega^2_{0,\dd}(M,\g^\ast)$. If $\dim(M)=2$, the statement $\zeta\in \Omega^2_{0,\dd}(M,\g^\ast)$
is automatic.
Applying $\delta$ on $\dd\psi_V = \square_{(2)}(\zeta)$ and using that $\delta\psi_V=0$ we find 
$\square_{(1)}(\psi_V) = \square_{(1)}(\delta\zeta)$ and hence $\psi_V = \delta \zeta$.
The condition $\tau(\EE^\mathrm{inv},\psi)\subseteq 2\pi\,\bbZ$ then reads
as in (\ref{eqn:tmpNull}), which implies that $[h^{-1}(G_{(1)}(\psi_V))]
\in 2\pi\, H^1_\mathrm{dR}(M,\g_\bbZ) $.
\end{proof}
With this characterization it is easy to see that the equation of motion is contained in $\NN$. 
\begin{lem}\label{lem:EOM}
$\MW^\ast\big[\Omega^1_0(M,\g^\ast)\big]\subseteq \NN$, where $\MW^\ast = \underline{\mathcal{F}}^\ast
 \circ \dd: \Omega^1_0(M,\g^\ast)\to \EE^\mathrm{kin}$ is the formal adjoint
of Maxwell's affine differential operator $\MW:= \delta\circ \underline{\mathcal{F}} : \sect{M}{\mathcal{C}(\Xi)} \to \Omega^1(M,\g)$.
\end{lem}
\begin{proof}
For any $\zeta \in \Omega^1_0(M,\g^\ast)$, $\MW^\ast(\zeta)_V = -\delta\dd \zeta$.
As a consequence of this and (\ref{eqn:commuteddeltaG}), 
$G_{(1)}(\MW^\ast(\zeta)_V) = - G_{(1)}(\delta\dd\zeta) = - G_{(1)}((\square_{(1)}-\dd\delta)(\zeta)) =
 \dd\delta G_{(1)}(\zeta)$
and thus $[h^{-1}(G_{(1)}(\MW^\ast(\zeta)_V))]=0$.
\end{proof}

The characterization of $\NN$ given in Proposition \ref{propo:null1} is still rather abstract. In particular,
it is quite hard to control the second condition since it involves the causal propagator and hence
the equation of motion together with its solution theory. Fortunately, it will be sufficient for us
to characterize explicitly only the  subgroup of $\NN$ given by
\begin{flalign}\label{eqn:N0}
\NN^0 := \big\{\psi\in \EE^\mathrm{inv}: \psi_V \in \delta\Omega^2_{0,\dd}(M,\g^\ast) ~\text{and}~
\big[h^{-1}\big(G_{(1)}(\psi_V)\big)\big]=0\big\}\subseteq \NN~.
\end{flalign}
Notice that $\NN^0$ can be defined as the set of all $\psi\in \EE^\mathrm{inv}$
 satisfying $\tau(\EE^\mathrm{inv},\psi) =\{0\}$, i.e.\ $\NN^0$ is the radical of the presymplectic structure 
 in $\EE^\mathrm{inv}$.
\begin{propo}\label{propo:Nmin}
$\NN^0 = \big\{\psi\in \EE^\mathrm{inv} : \psi_V 
\in \delta\big(\Omega^2_0(M,\g^\ast)\cap \dd \Omega^1_\mathrm{tc}(M,\g^\ast)\big)\big\}$, where the subscript
$_\mathrm{tc}$ stands for forms with timelike compact support.
\end{propo}
\begin{proof}
We first show the inclusion $\subseteq$: 
Let $\psi_V=\delta\zeta$, $\zeta \in \Omega^2_{0,\dd}(M,\g^\ast)$, be the linear part of $\psi\in\NN^0$.
The second condition in (\ref{eqn:N0}) implies that there exists a $\chi^\prime \in C^\infty(M,\g)$, such that
$h^{-1}(G_{(1)}(\psi_V)) = \dd \chi^\prime$. Absorbing $h^{-1}$ into $\chi^\prime$ we obtain 
the equivalent equation $G_{(1)}(\psi_V) = \dd\chi$, for some $\chi\in C^\infty(M,\g^\ast)$. Applying $\delta$ to both sides
leads to $\square_{(0)}(\chi)=0$, hence there exists an $\alpha\in C^\infty_{\mathrm{tc}}(M,\g^\ast)$ such that
$\chi = G_{(0)}(\alpha)$.
See e.g.\ \cite{Baernew,Sanders:2012ac}  for details on how to extend the causal propagator
to sections of timelike compact support. The original equation $G_{(1)}(\psi_V) = \dd\chi$ implies that
$\psi_V  = \dd\alpha + \square_{(1)}(\beta) $ for some $\beta\in \Omega^1_{\mathrm{tc}}(M,\g^\ast)$.
Applying $\delta$ and using that $\delta\psi_V=0$ gives $\alpha = -\delta \beta$ and the equation simplifies to
$\delta\zeta = \psi_V = \delta\dd \beta$. Applying $\dd$ and using $\dd\zeta =0$ 
shows that $\zeta = \dd\beta$ and hence $\dd\beta \in\Omega^2_0(M,\g^\ast)$.
The other inclusion $\supseteq$ is easily shown, for all $\dd \beta \in
 \Omega^2_0(M,\g^\ast)\cap \dd \Omega^1_\mathrm{tc}(M,\g^\ast)$,
\begin{flalign}
G_{(1)}(\delta\dd\beta) = \delta\dd G_{(1)}(\beta) = (\square_{(1)} - \dd\delta)\big(G_{(1)}(\beta)\big)
= -\dd\delta G_{(1)}(\beta)
 \end{flalign}
 and hence $[h^{-1}(G_{(1)}(\delta\dd\beta))]=0$.
\end{proof}
\begin{lem}\label{lem:EOMN0}
$\MW^\ast[\Omega^1_0(M,\g^\ast)]\subseteq \NN^0$.
\end{lem}
\begin{proof}
Follows immediately from the proof of Lemma \ref{lem:EOM}.
\end{proof}


\section{\label{sec:functor} The presymplectic Abelian group functor and its quotients} 
We associate the presymplectic Abelian groups constructed in the previous section
to objects $\Xi$ in the category $G{-}\PrBu$ and study how morphisms
in $G{-}\PrBu$ induce morphisms between these presymplectic Abelian groups. Our strategy
is to construct first an off-shell functor, i.e.\ a functor which does not encode the equation 
of motion, and afterwards we shall study the possibility of taking natural quotients of
this functor using a more abstract mathematical machinery. This point of view will be useful
for establishing certain properties of our functors.
For the definition of the category $\PAG$ of presymplectic Abelian groups
we refer to the Appendix, Definition \ref{defi:PAG}.
\begin{propo}\label{propo:offshellfunctor}
The following association defines a covariant functor $\PhaseSpaceOff: G{-}\PrBu \to \PAG$:  
For objects $\Xi$ in $G{-}\PrBu$ we set  $\PhaseSpaceOff(\Xi) := (\EE^\mathrm{inv},\tau)$,
where $\EE^\mathrm{inv}$ is given in (\ref{eqn:Einv}) and $\tau$ in (\ref{eqn:taumap}).
For morphisms $f:\Xi_1\to \Xi_2$  in $G{-}\PrBu$ we set
\begin{flalign}\label{eqn:funcmorph}
\PhaseSpaceOff(f): \PhaseSpaceOff(\Xi_1)\to \PhaseSpaceOff(\Xi_2) ~,~~\varphi \mapsto f_\ast(\varphi)~,
\end{flalign}
where $f_\ast$ is the push-forward given in \cite[Definition 5.4]{Benini:2013tra}.
\end{propo}
\begin{proof}
The proof can be obtained by following the same steps as in the proof of \cite[Theorem 5.5]{Benini:2013tra}. 
As the reader might ask if the integral cohomology condition in the definition of $\EE^\mathrm{inv}$ in
(\ref{eqn:Einv}) can cause problems, we are repeating the relevant part of this proof.
The only non-trivial step is to show that the morphisms (\ref{eqn:funcmorph}) are well-defined. 
Using the short notation $(\EE^\mathrm{inv}_1 ,\tau_1) := \PhaseSpaceOff(\Xi_1)$
and $(\EE^\mathrm{inv}_2,\tau_2):= \PhaseSpaceOff(\Xi_2)$, this amounts to showing that the push-forward
$f_\ast$ maps $\EE^\mathrm{inv}_1$ to $\EE^\mathrm{inv}_2$. (The proof that $f_\ast$ preserves the presymplectic
structures is exactly the one in \cite[Theorem 5.5]{Benini:2013tra}.) This is indeed the case, since, for
all $\varphi\in \EE_1^\mathrm{inv}$ and $\widehat{g}\in C^\infty(M_2,G)$,
\begin{flalign}
\ip{f_\ast(\varphi)_V}{ \widehat{g}^\ast(\mu_G)}_2 = \langle\underline{f}_\ast(\varphi_V),\;\widehat{g}^\ast(\mu_G)\rangle_2 = \ip{\varphi_V}{\underline{f}^\ast\circ \widehat{g}^\ast(\mu_G)}_1 = \ip{\varphi_V}{(\widehat{g}\circ \underline{f})^\ast(\mu_G)}_1=0~,
\end{flalign}
where $\underline{f}: M_1\to M_2$ is the map induced by $f:P_1\to P_2$ and we have used that 
$\widehat{g}\circ \underline{f}\in C^\infty(M_1,G)$. Hence, $f_\ast(\varphi) \in \EE^\mathrm{inv}_2$
 for all $\varphi\in \EE^\mathrm{inv}_1$.
\end{proof}
\begin{rem}
The covariant functor $\PhaseSpaceOff: G{-}\PrBu \to \PAG$ restricts in the obvious way
to the full subcategory $G{-}\PrBu^{(m)}$ (and also to $G{-}\PrBu^{(m)}_0$) of $G{-}\PrBu$ (with $m\geq 2$),
which describes principal $G$-bundles over (connected) $m$-dimensional spacetimes for a fixed $m\geq 2$.
We shall denote the restricted functors by the same symbol, i.e.\ $\PhaseSpaceOff: G{-}\PrBu^{(m)} \to \PAG$
and $\PhaseSpaceOff: G{-}\PrBu^{(m)}_0 \to \PAG$.
\end{rem}

The covariant functor $\PhaseSpaceOff$ is not yet the one required in physics since it does not encode the equation of motion.
We will address the question of taking quotients of the objects $\PhaseSpaceOff(\Xi)$
by subgroups $\SubGr(\Xi)\subseteq \PhaseSpaceOff(\Xi)$ from a more abstract point of view.
This is required to understand if we can take in our present model
a quotient by the equation of motion and also certain ``electric charges'', cf.~\cite[Section 7]{Benini:2013tra}.
Eventually, this will decide whether the covariant functor resulting from taking quotients satisfies the locality
property (i.e.~injectivity of the induced morphisms in $\PAG$) or not.

There are the following restrictions on the choice of the collection
 $\SubGr(\Xi)\subseteq \PhaseSpaceOff(\Xi)$ of subgroups:
 First, for $\PhaseSpaceOff(\Xi)/\SubGr(\Xi)$ to be an object in $\PAG$ (with the induced presymplectic structure) 
 it is necessary and sufficient that  $\SubGr(\Xi)$ is a subgroup of the radical $\NN^0\subseteq \PhaseSpaceOff(\Xi)$.
 Second, for $\PhaseSpaceOff(f):\PhaseSpaceOff(\Xi_1)\to \PhaseSpaceOff(\Xi_2)$ to induce a morphism on the quotients it is
 necessary and sufficient that $\PhaseSpaceOff(f)$ maps $\SubGr(\Xi_1)$ to $\SubGr(\Xi_2)$.
 These conditions can be abstractly phrased as follows.
 \begin{defi}\label{defi:quotient}
 Let $\mathsf{C}$ be any category and let $\mathfrak{F}:\mathsf{C} \to \PAG$ be a covariant functor.
 \begin{itemize}
 \item[a)] A covariant functor $\SubGr:\mathsf{C}\to \PAG$ is called a {\bf subfunctor} of $\mathfrak{F}$
 if for all objects $A$ in $\mathsf{C}$ we have $\SubGr(A)\subseteq \mathfrak{F}(A)$ (i.e.\  $\SubGr(A)$
 is a presymplectic Abelian subgroup of $\mathfrak{F}(A)$) and if for all morphisms $f:A_1\to A_2$ in $\mathsf{C}$
 the morphism $\SubGr(f)$ is the restriction of $\mathfrak{F}(f)$ to $\SubGr(A_1)$.
 \item[b)] A subfunctor $\SubGr:\mathsf{C}\to \PAG$ of $\mathfrak{F}$ is called {\bf quotientable} if
 for all objects $A$ in $\mathsf{C}$ the presymplectic Abelian group $\SubGr(A)$ is a presymplectic
  Abelian subgroup of the radical in $\mathfrak{F}(A)$.
 \end{itemize}
 \end{defi}
 \begin{rem}
 Notice that if the category $\mathsf{C}$ is $G{-}\PrBu$ and the functor $\mathfrak{F}$ is $\PhaseSpaceOff$
 we  recover exactly the situation explained before Definition \ref{defi:quotient}. 
 We have formulated the definition in this generality, since 
 we shall also encounter the case of a category $\mathsf{C}$ different from $G{-}\PrBu$.
 \end{rem}
\begin{propo}\label{propo:quotientphasespace}
Let  $\SubGr: \mathsf{C} \to \PAG$ be a quotientable subfunctor of a 
covariant functor $\mathfrak{F}:\mathsf{C} \to \PAG$. Then there exists
 a covariant functor $\mathfrak{F}/\SubGr : \mathsf{C}\to \PAG$, called the quotient of $\mathfrak{F}$ by $\SubGr$, 
 defined as follows: It associates to any object 
$A$ in $\mathsf{C}$ the object $\mathfrak{F}(A)/\SubGr(A)$ in $\PAG$. To any
morphism $f:A_1\to A_2$  in $\mathsf{C}$ the functor associates the morphism
$\mathfrak{F}(A_1)/\SubGr(A_1)\to \mathfrak{F}(A_2)/\SubGr(A_2)$ in 
$\PAG$ that is canonically induced by $\mathfrak{F}(f)$.
\end{propo}
\begin{proof}
For any object $A$ in $\mathsf{C}$ the quotient $\mathfrak{F}(A)/\SubGr(A)$ is an object in $\PAG$, since
$\SubGr(A) $ is a presymplectic Abelian subgroup of the radical in $\mathfrak{F}(A)$.
For any morphism $f:A_1\to A_2$ in $\mathsf{C}$
the morphism $\mathfrak{F}(f): \mathfrak{F}(A_1)\to \mathfrak{F}(A_2)$ induces
a well-defined morphism between the quotients, since by the subfunctor properties
$\SubGr(A_1)$ is mapped to $\SubGr(A_2)$.
\end{proof}
It remains to provide explicit examples of quotientable subfunctors
of $\PhaseSpaceOff: G{-}\PrBu \to \PAG$. The following example is standard, since it describes within the terminology
developed above the quotient by the equation of motion.
\begin{propo}\label{propo:MWquotient}
Let $\PhaseSpaceOff:G{-}\PrBu \to \PAG$ be the functor constructed in Proposition \ref{propo:offshellfunctor}.
Then there exists a quotientable subfunctor $\MWfunc: G{-}\PrBu\to \PAG$ defined by associating
to any object $\Xi$ in $G{-}\PrBu$ the presymplectic
Abelian subgroup $\MWfunc(\Xi) := (\MW^\ast[\Omega^1_0(M,\g^\ast)], \tau) \subseteq \PhaseSpaceOff(\Xi)$, where
$\tau$ is given in  (\ref{eqn:taumap}).
\end{propo}
\begin{proof}
Since $(\MW^\ast[\Omega^1_0(M,\g^\ast)], \tau)$ is clearly a presymplectic Abelian subgroup
of $\PhaseSpaceOff(\Xi)$, for all objects $\Xi$ in $G{-}\PrBu$, it remains to check
if the morphism $\PhaseSpaceOff(f)$ defined in (\ref{eqn:funcmorph}) induces a morphism
between $\MW_1^\ast[\Omega^1_0(M_1,\g^\ast)]$ and $\MW_2^\ast[\Omega^1_0(M_2,\g^\ast)]$.
This was shown in the proof of \cite[Theorem 5.5]{Benini:2013tra}, hence $\MWfunc$ is a subfunctor
of $\PhaseSpaceOff$. It is a quotientable subfunctor, since by Lemma \ref{lem:EOMN0} 
all $\MWfunc(\Xi)$ are presymplectic Abelian subgroups of the radical in $\PhaseSpaceOff(\Xi)$ 
(this also implies that $\tau$ restricted to $\MW^\ast[\Omega^1_0(M,\g^\ast)]$ is the trivial presymplectic
structure).
\end{proof}

Using Proposition \ref{propo:quotientphasespace} we construct a covariant
functor 
\begin{flalign}
\PhaseSpace := \PhaseSpaceOff/\MWfunc : G{-}\PrBu \to \PAG~.
\end{flalign}
This functor describes exactly the gauge invariant {\it on-shell} presymplectic Abelian groups, i.e.~for any object $\Xi$
in $G{-}\PrBu$ we have $\PhaseSpace(\Xi) = (\EE^\mathrm{inv}/\MW^\ast[\Omega^1_0(M,\g^\ast)],\tau)$
with $\EE^\mathrm{inv}$ given in (\ref{eqn:Einv}) and $\tau$ in (\ref{eqn:taumap}).
Of course, the functor $\PhaseSpace$ restricts to the full subcategories $G{-}\PrBu^{(m)}$ and
$G{-}\PrBu^{(m)}_0$, for any $m\geq 2$.
The functor $\PhaseSpace$ has many desired properties of a locally covariant field theory,
namely the causality property and the time-slice axiom, which can be shown using the same arguments
as in the proof of \cite[Theorem 5.7 and Theorem 5.8]{Benini:2013tra}.
However, it does not satisfy the locality property.
\begin{defi}
Let $\mathsf{C}$ be any category.
A covariant functor $\mathfrak{F}:\mathsf{C}\to \PAG$ is said to satisfy the {\bf locality property}, if
it is a functor to the subcategory $\PAG^\mathrm{inj}$ where all morphisms are injective,
cf.~Definition \ref{defi:PAG}.
\end{defi}

In order to prove that the functor $\PhaseSpace$ does not satisfy the locality property
we shall need some technical tools. Even though the next steps are rather abstract,
taking this burden will pay off, since we can eventually rephrase the injectivity of the morphism
$\PhaseSpace(f)$ in $\PAG$ in terms of an injectivity condition on a certain induced morphism of cohomology groups.

Let us define the following covariant functor $H^2_{0\,\mathrm{dR}}: G{-}\PrBu \to \Vec$ to the category of real vector spaces:
To any object $\Xi$ in $G{-}\PrBu$ we associate the cohomology group 
$H^2_{0\,\mathrm{dR}}(\Xi):= H^2_{0\,\mathrm{dR}}(M,\g^\ast)$ 
of the spacetime $M$ underlying $\Xi$. To any morphism $f:\Xi_1\to \Xi_2$ in $G{-}\PrBu$ we associate
 the linear map
\begin{flalign}\label{eqn:tmpcohomologymap}
H^2_{0\,\mathrm{dR}}(f) :  H^2_{0\,\mathrm{dR}}(M_1,\g^\ast)
\to H^2_{0\,\mathrm{dR}}(M_2,\g^\ast)  ~,~~[\zeta]\mapsto  [\underline{f}_\ast(\zeta)]~,
\end{flalign}
where $\underline{f}_\ast$ is the push-forward along the induced map $\underline{f} : M_1\to M_2$.
We shall compose both, the functor $\PhaseSpace$ and the functor $H^2_{0\,\mathrm{dR}}$, with the forgetful
functor to the category of Abelian groups $\mathsf{AG}$ and denote the resulting functor with a  slight abuse of
notation by the same symbol. We observe
\begin{lem}\label{lem:nattrafocoho}
Let us define for every object $\Xi$ in $G{-}\PrBu$ a morphism of Abelian groups
\begin{flalign}\label{eqn:tmpcomap2}
\iota_\Xi : H^2_{0\,\mathrm{dR}}(\Xi) \to \PhaseSpace(\Xi)~,~~[\zeta] \mapsto \big[\underline{\mathcal{F}}^\ast(\zeta)\big]~.
\end{flalign}
Then the collection $\iota = \{\iota_\Xi\} : H^2_{0\,\mathrm{dR}}\Rightarrow \PhaseSpace$ defines a natural transformation
between the two covariant functors $\PhaseSpace, H^2_{0\,\mathrm{dR}}:G{-}\PrBu\to\mathsf{AG}$.
Furthermore, each $\iota_\Xi$ is injective.
\end{lem}
\begin{proof}
Let $\Xi$ be any object in $G{-}\PrBu$. Then the map (\ref{eqn:tmpcomap2}) is well-defined, 
since, for any $\zeta\in\Omega^2_{0,\dd}(M,\g^\ast)$, $\underline{\mathcal{F}}^\ast(\zeta)_V 
=-\delta \zeta \in \EE^\mathrm{inv}$, cf.\ (\ref{eqn:Einv}), 
and since, for any $\eta \in \Omega^1_0(M,\g^\ast)$, 
$\big[\underline{\mathcal{F}}^\ast(\dd\eta)\big] = \big[\MW^\ast(\eta)\big] =0$.
It is clearly a morphism of Abelian groups.
To show that $\iota_\Xi$ is injective, let us assume that $\big[\underline{\mathcal{F}}^\ast(\zeta)\big]=0$. Hence,
there exists an $\eta\in \Omega^1_0(M,\g^\ast)$, such that $\underline{\mathcal{F}}^\ast(\zeta)=\MW^\ast(\eta)$.
Taking the linear part gives $\delta\zeta = \delta\dd\eta$. Applying $\dd$ and
using that $\dd\zeta=0$ we obtain $\zeta=\dd\eta$, thus $[\zeta]=0$.

Let now $f:\Xi_1\to \Xi_2$ be any morphism in $G{-}\PrBu$. It remains to show that the diagram
\begin{flalign}\label{eqn:diagcomap}
\xymatrix{
\PhaseSpace(\Xi_1) \ar[rr]^-{\PhaseSpace(f)} && \PhaseSpace(\Xi_2)\\
H^2_{0\,\mathrm{dR}}(\Xi_1) \ar[u]^-{\iota_{\Xi_1}}\ar[rr]^-{H^2_{0\,\mathrm{dR}}(f)} && H^2_{0\,\mathrm{dR}}(\Xi_2) \ar[u]_-{\iota_{\Xi_2}}
}
\end{flalign}
commutes. This is a simple calculation, for all $[\zeta]\in H^2_{0\,\mathrm{dR}}(\Xi_1)$,
\begin{flalign}
\PhaseSpace(f)\big(\iota_{\Xi_1}\big([\zeta]\big)\big)
=\big[f_\ast\big(\underline{\mathcal{F}}_1^\ast(\zeta)\big)\big]
=\big[\underline{\mathcal{F}}_2^\ast\big(\underline{f}_\ast(\zeta)\big)\big]
=\iota_{\Xi_2}\big(H^2_{0\,\mathrm{dR}}(f)\big([\zeta]\big)\big)~,
\end{flalign}
where in the second equality we have used naturality of the curvature affine differential operator
 -- see \cite[Lemma 2.14]{Benini:2013tra} and the subsequent discussion.
\end{proof}
With this preparation we can give a simple characterization of the injectivity of the morphism
$\PhaseSpace(f)$ in terms of injectivity of $H^2_{0\,\mathrm{dR}}(f)$.
\begin{theo}\label{theo:diagcomap}
 Let $f: \Xi_1\to\Xi_2$ be a morphism in $G{-}\PrBu$.
 Then $\PhaseSpace(f): \PhaseSpace(\Xi_1)\to\PhaseSpace(\Xi_2)$ is injective if and only if
 $H^2_{0\,\mathrm{dR}} (f): H^2_{0\,\mathrm{dR}}(\Xi_1)\to H^2_{0\,\mathrm{dR}}(\Xi_2)$
 is injective. In other words, $\PhaseSpace(f): \PhaseSpace(\Xi_1)\to\PhaseSpace(\Xi_2)$ is not injective if and only if
  $H^2_{0\,\mathrm{dR}} (f): H^2_{0\,\mathrm{dR}}(\Xi_1)\to H^2_{0\,\mathrm{dR}}(\Xi_2)$
  is not injective. 
\end{theo}
\begin{proof}
It is easier to show the negation of the statement, i.e.\ $\PhaseSpace(f): \PhaseSpace(\Xi_1)\to\PhaseSpace(\Xi_2)$ 
is not injective $\Leftrightarrow$  $H^2_{0\,\mathrm{dR}} (f): H^2_{0\,\mathrm{dR}}(\Xi_1)\to H^2_{0\,\mathrm{dR}}(\Xi_2)$ 
is not injective.

The direction ``$\Leftarrow$'' follows immediately from the commuting diagram in (\ref{eqn:diagcomap}). Indeed,
if $H^2_{0\,\mathrm{dR}} (f)$ is not injective, then the lower composition of morphisms is not injective
and by commutativity of the diagram also the upper composition of morphisms is not injective. Since
$\iota_{\Xi_1}$ is injective this implies that $\PhaseSpace(f)$ is not injective.

To show the direction ``$\Rightarrow$'' let us assume that $\PhaseSpace(f)$
is not injective. The kernel $\ker\big(\PhaseSpace(f)\big)$ is a subgroup of the radical $\big[\NN^0_1\big] := \NN^0_1/\MWfunc(\Xi_1)$
in $\PhaseSpace(\Xi_1)$ for the following reason: 
For any non-trivial element $0\neq [\psi]\in \ker\big(\PhaseSpace(f)\big)$
we have (by definition) that $\PhaseSpace(f)([\psi]) =0$ and hence, for all $[\varphi]\in \PhaseSpace(\Xi_2)$,
a vanishing presymplectic structure $\tau_2\big([\varphi],\PhaseSpace(f)([\psi])\big)=0$. Taking in particular
 $[\varphi] = \PhaseSpace(f)([\widetilde\varphi])$,
for $[\widetilde\varphi]\in \PhaseSpace(\Xi_1)$, and using that $\PhaseSpace(f)$ preserves the presymplectic
structures implies that $[\psi]$ lies in the radical $\big[\NN^0_1\big]$, which by Proposition \ref{propo:Nmin} implies
that for any representative $\psi$ we have $\psi_V = -\delta \dd\beta$ (the minus sign is purely conventional) for some
 $\beta \in \Omega^1_{\mathrm{tc}}(M_1,\g^\ast)$.
This in turn implies that there exists $a\in C_0^\infty(M_1)$, such that 
$\psi = \underline{\mathcal{F}}^\ast_1(\dd\beta) + a\,\1_1$,
and since $\PhaseSpace(f)([\psi])=0$, the push-forward $f_\ast(\psi)$ has to be of the form
$\MW_2^\ast(\eta)$ for some $\eta\in \Omega^1_0(M_2,\g^\ast)$, i.e.\ 
\begin{flalign}\label{eqn:tmp12345}
 \MW_2^\ast(\eta) = f_\ast(\psi) =\underline{\mathcal{F}}^\ast_2\big(\underline{f}_\ast(\dd\beta)\big) + \underline{f}_\ast(a) \,\1_2 ~.  
 \end{flalign}
 Taking the linear part of this equation, i.e.\ $-\delta\dd\eta = - \delta\underline{f}_\ast(\dd\beta)$,
 and applying $\dd$ to both sides leads to $\underline{f}_\ast(\dd\beta) = \dd\eta$ and hence, by plugging this back into 
 (\ref{eqn:tmp12345}), we obtain that $\underline{f}_\ast(a)\in\mathrm{Triv}_{\bbZ\,2}$, which is equivalent to
 $a\in \mathrm{Triv}_{\bbZ\,1}$. Hence, $\psi = \underline{\mathcal{F}}^\ast_1(\dd\beta)  $
 and as a consequence $[\psi]$ lies in the image of $\iota_{\Xi_1}$. This shows that the upper compositions
 of morphisms in (\ref{eqn:diagcomap}) is not injective, hence (as a consequence of the commutativity of the diagram) also
 the lower composition is not injective. As $\iota_{\Xi_2}$ is injective, the morphism 
 $H^2_{0\,\mathrm{dR}}(f)$ has to be non-injective, which proves our claim.
\end{proof}

\begin{propo}\label{propo:nonlocal}
The covariant functor  $\PhaseSpace  : G{-}\PrBu^{(m)} \to \PAG$ does not satisfy the locality property,
for any $m\geq 2$. Furthermore, the covariant functor $\PhaseSpace : G{-}\PrBu^{(m)}_0 \to \PAG$
(i.e.\ the restriction to connected spacetimes of dimension $m$) does not satisfy the locality property, for any $m\geq 3$.
\end{propo}
\begin{rem}
This proposition implies that also the covariant functor 
$\PhaseSpace : G{-}\PrBu \to \PAG$ on the 
category $G{-}\PrBu$ does not satisfy the locality property.
\end{rem}
\begin{proof}
By Theorem \ref{theo:diagcomap}, it is enough to construct examples of morphisms $f:\Xi_1\to\Xi_2$ 
in $G{-}\PrBu^{(m)}$ (for $m\geq 2$) and also in $G{-}\PrBu^{(m)}_0$ (for $m\geq 3$),
such that $H^2_{0\,\mathrm{dR}}(f): H^2_{0\,\mathrm{dR}}(M_1,\g^\ast) \to H^2_{0\,\mathrm{dR}}(M_2,\g^\ast) $
is not injective. A sufficient condition for this property is that
the vector space dimension of $H^2_{0\,\mathrm{dR}}(M_1,\g^\ast) $ is greater than the dimension of
$H^2_{0\,\mathrm{dR}}(M_2,\g^\ast) $. This is achieved by the following construction:
For $m\geq 2$, let $\Xi_2$ be any object in $G{-}\PrBu^{(m)}$, such that $(M_2,\o_2,g_2,\t_2)$ 
is the $m$-dimensional Minkowski spacetime. Let us further denote by $\Xi_1$ the object in $G{-}\PrBu^{(m)}$
that is obtained by restricting all geometric data of $\Xi_2$ to the causally compatible and globally hyperbolic
open subset $M_1:= M_2\setminus J_{M_2}(\{0\})$, where $0\in M_2$ is some point in $M_2$. The canonical
embedding of $M_1$ into $M_2$ provides us with a morphism $f:\Xi_1\to\Xi_2$ in $G{-}\PrBu^{(m)}$.
Notice that for $m\geq 3$ this $f$ is also a morphism in $G{-}\PrBu^{(m)}_0$, as both, $M_1$ and $M_2$, 
are connected in this case.
Let us study the dimension of the de Rham cohomology groups $H^2_{0\,\mathrm{dR}}(M_1,\g^\ast)$
and $H^2_{0\,\mathrm{dR}}(M_2,\g^\ast)$. In dimension $m=2$ we have that $H^2_{0\,\mathrm{dR}}(M_1,\g^\ast)\simeq \bbR^2$
and $H^2_{0\,\mathrm{dR}}(M_2,\g^\ast)\simeq \bbR$, since $M_1$ consists of two disconnected components
and $M_2$ is connected. Hence, the morphism 
$H^2_{0\,\mathrm{dR}}(f)$ can not be injective and by Theorem \ref{theo:diagcomap} we find that $\PhaseSpace(f)$ is not injective
for this choice of morphism $f:\Xi_1\to\Xi_2$. In dimension $m>2$ we have that 
$H^2_{0\,\mathrm{dR}}(M_1,\g^\ast)\simeq \bbR$ (since $M_1$ is diffeomorphic to $\bbR^2\times \mathbb{S}^{m-2}$
with $\mathbb{S}^{m-2}$ denoting the $m-2$-sphere) and $H^2_{0\,\mathrm{dR}}(M_2,\g^\ast)=\{0\}$, thus again
$\PhaseSpace(f)$ is not injective. 
\end{proof}
This proposition shows that the usual on-shell functor $\PhaseSpace : G{-}\PrBu\to \PAG$,
as well as its restriction to the full subcategories $G{-}\PrBu^{(m)}$ (for $m\geq 2$) and
$G{-}\PrBu^{(m)}_0$ (for $m\geq 3$), is not locally covariant in the sense of \cite{Brunetti:2001dx}.
An exceptional case is given for the full subcategory $G{-}\PrBu^{(2)}_0$, where each spacetime
is two-dimensional and connected.  In this case local covariance in the sense of \cite{Brunetti:2001dx} holds true.
\begin{theo}
The covariant functor $\PhaseSpace: G{-}\PrBu^{(2)}_0 \to \PAG$ satisfies the locality property. 
\end{theo}
\begin{proof}
Let $f:\Xi_1\to\Xi_2$ be any morphism in $G{-}\PrBu^{(2)}_0$ and let us denote by
$H^2_{0\,\mathrm{dR}}(f): H^2_{0\,\mathrm{dR}}(M_1,\g^\ast)\to H^2_{0\,\mathrm{dR}}(M_2,\g^\ast)$ 
the induced morphism of cohomology groups. 
Since by assumption $M_1$ and $M_2$ are connected, both $H^2_{0\,\mathrm{dR}}(M_1,\g^\ast)$ 
and $H^2_{0\,\mathrm{dR}}(M_2,\g^\ast)$ are isomorphic to $\mathbb{R}$ by Poincar\'e duality. 
Let $[\zeta]\in H^2_{0\,\mathrm{dR}}(M_1,\g^\ast)$ be a non-trivial cohomology class. 
Then $H^2_{0\,\mathrm{dR}}(f)([\zeta])=[\underline{f}_\ast (\zeta) ]$ is not trivial, since 
$\int_{M_2}\underline{f}_\ast(\zeta) =\int_{M_1}\zeta\neq 0$. 
In other words, $H^2_{0\,\mathrm{dR}}(f)$ is injective and by 
Theorem \ref{theo:diagcomap} also $\PhaseSpace(f)$. This completes the proof.
\end{proof}

In order to circumvent the violation of the locality property of the functor $\PhaseSpace$ acting on the categories
$G{-}\PrBu^{(m)}$ (with $m\geq 2$) or $G{-}\PrBu^{(m)}_0$ (with $m\geq 3$), 
there remains the possibility of taking further quotients of the functor $\PhaseSpace$ by quotientable subfunctors,
cf.~\cite[Section 7]{Benini:2013tra} for a similar strategy.
However, this turns out to be impossible due to the following
\begin{theo}\label{theo:noinjectivequotient}
For any $m\geq 2$ there exists no quotientable subfunctor  $\SubGr$ of $\PhaseSpace: G{-}\PrBu^{(m)} \to \PAG$, such that
$\PhaseSpace/\SubGr: G{-}\PrBu^{(m)}\to \PAG$ satisfies the locality property.
Furthermore, for any $m\geq 3$ there exists no quotientable subfunctor $\SubGr$ of
$\PhaseSpace:G{-}\PrBu^{(m)}_0\to \PAG$, such that $\PhaseSpace/\SubGr: G{-}\PrBu^{(m)}_0 \to \PAG$ satisfies the locality property.
\end{theo}
\begin{rem}
This theorem of course implies that there exists no quotientable subfunctor  $\SubGr$ of $\PhaseSpace: G{-}\PrBu \to \PAG$, such that
$\PhaseSpace/\SubGr: G{-}\PrBu\to \PAG$ satisfies the locality property.
\end{rem}
\begin{proof}
The strategy for the proof is as follows: We will construct 
two morphisms $f_j : \Xi_3\to \Xi_j$, $j=1,2$, in $G{-}\PrBu^{(m)}$ (for $m\geq 2$) and $G{-}\PrBu^{(m)}_0$  (for $m\geq 3$),
 such that injectivity  of $\PhaseSpace(f_2)$ on the quotients requires $\SubGr(\Xi_3)$ in such a way that
$\PhaseSpace(f_1)$ is not well-defined on the quotients. The following diagram visualizes the envisaged setting:
\begin{flalign}
\xymatrix{
 \Xi_1 & & \Xi_2\\
 & \ar[lu]^-{f_1}\Xi_3 \ar[ru]_-{f_2}&
}
\end{flalign}

Let us first focus on the case $m\geq 3$ and consider $M_1 := \bbR \times \mathbb{S}^1\times \mathbb{S}^{m-2}$ 
equipped with the canonical Lorentzian metric
$g_1 := -dt\otimes dt + d\phi\otimes d\phi + g_{\mathbb{S}^{m-2}}$, where $t$ is a time coordinate
on $\bbR$,  $\phi$ is an angle coordinate on the circle $\mathbb{S}^{1}$ and $g_{\mathbb{S}^{m-2}}$ is the 
standard Riemannian metric on the unit $m-2$-sphere $\mathbb{S}^{ m-2}$.
Consider further $M_2:= \bbR^m$ equipped with the Lorentzian metric
$g_2:= -dt\otimes dt + \alpha(r^2)\,\sum_{i=1}^{m-1} dx^i\otimes dx^i
+ \beta(r^2)\,dr\otimes dr$, where $x^i$ are Cartesian coordinates, $r=\sqrt{\sum_{i=1}^{m-1}x^i\,x^i}$ is the radius,
$\alpha : \bbR \to\bbR$ is a strictly positive smooth function, such that $\alpha(\xi) = 1$ for $\xi<1$
and $\alpha(\xi) = \xi^{-1}$ for $\xi>4$, and $\beta :\bbR \to \bbR$ is a positive
smooth function, such that $\beta(\xi)=0$ for  $\xi<1$ and $\beta(\xi) = 1- \xi^{-1}$ for $\xi>4$.
Notice that, for $r^2<1$, $g_2 = -dt\otimes dt +\sum_{i=1}^{m-1} dx^i\otimes dx^i $ is the Minkowski metric 
 and that, for $r^2>4$, $g_2 = -dt\otimes dt + dr\otimes dr + g_{\mathbb{S}^{m-2}}$ formally looks like
 $g_1$.

 Define $M_3$ as the Cauchy development 
 of $\{0\} \times I \times \mathbb{S}^{m-2}$ in $(M_1,g_1)$,
 where $I$ is some open interval in $(0,2\pi)$. Notice that $(M_3,g_3 := g_1\vert_{M_3})$ is a causally compatible 
 and globally hyperbolic open subset of $(M_1,g_1)$. We denote by $\underline{f_1}:(M_3,g_3) \to (M_1,g_1)$ the isometric embedding.
 Furthermore, there exists  an isometric embedding $\underline{f_2}:(M_3,g_3) \to (M_2,g_2)$ into the 
 subset of $M_2$ specified by $r^2>4$, such that
 the image is causally compatible and open.\footnote{
 Explicitly, we can define an isometric embedding $M_1\supset \bbR \times I \times \mathbb{S}^{m-2}\to M_2$
 by $(t,\phi,p) \mapsto (t, 2+\phi,p)$, where $p$ is a point on the sphere $\mathbb{S}^{m-2}$ and 
  the coordinates $(t,r,p)$ on $M_2 = \bbR^m$ have been chosen as time $t$ and
  spherical coordinates $(r,p)$ on the equal-time hypersurfaces.
 This induces the desired isometric embedding $\underline{f_2}:(M_3,g_3) \to (M_2,g_2)$.
 } 
 We can equip $M_i$, $i=1,2,3$, with orientations and time-orientations,
 such that the isometric embeddings preserve those. Furthermore,
 taking the trivial principal $G$-bundles $P_i = M_i \times G$, $i=1,2,3$, 
 we can construct three objects $\Xi_i$, $i=1,2,3$, and two morphisms
 $f_j :=(\underline{f_j},\id_G): M_3\times G \to M_j \times G$, $j=1,2$,
 in $G{-}\PrBu_0^{(m)}$.

By the same arguments as in the proof of Proposition $\ref{propo:nonlocal}$,
the morphism $\PhaseSpace(f_2)$ can not be injective since $H^2_{0\,\mathrm{dR}}(M_3,\g^\ast) \simeq \bbR$ and
$H^2_{0\,\mathrm{dR}}(M_2,\g^\ast) =\{0\}$. To turn $\PhaseSpace(f_2): \PhaseSpace(\Xi_3)\to \PhaseSpace(\Xi_2)$
into an injective morphism we have to take a quotient $\PhaseSpace(\Xi_3)/\SubGr(\Xi_3)$, such that
$\SubGr(\Xi_3)$ contains $[\underline{\mathcal{F}}^\ast_3(\zeta)]$, for all $\zeta\in\Omega^2_{0,\dd}(M_3,\g^\ast)$,
as all these elements lie in the kernel of $\PhaseSpace(f_2)$. (The last statement can also be seen very explicitly:
For all $\zeta\in\Omega^2_{0,\dd}(M_3,\g^\ast)$ we have $\PhaseSpace(f_2)\big([\underline{\mathcal{F}}^\ast_3(\zeta)]\big)
= [\underline{\mathcal{F}}^\ast_2(\underline{f_2}_\ast(\zeta))] = [\MW_2(\eta)]=0$, since 
$H^2_{0\,\mathrm{dR}}(M_2,\g^\ast)=\{0\}$ and hence $\underline{f_2}_\ast(\zeta) = \dd \eta$ 
for some $\eta\in \Omega^1_0(M_2,\g^\ast)$.)
On the other hand, the morphism 
$H^2_{0\,\mathrm{dR}}(f_1): H^2_{0\,\mathrm{dR}}(M_3,\g^\ast) \to H^2_{0\,\mathrm{dR}}(M_1,\g^\ast)$ 
is injective, since for any non-trivial $0\neq [\zeta]\in H^2_{0\,\mathrm{dR}}(M_3,\g^\ast)$
we have $\int_{M_1} \underline{f_1}_\ast(\zeta)\wedge \nu_{\mathbb{S}^{m-2}}^1 
= \int_{M_3} \zeta\wedge \nu_{\mathbb{S}^{m-2}}^3\neq 0$, where $\nu_{\mathbb{S}^{m-2}}^i$ denotes the pull-back of the
volume form on $\mathbb{S}^{m-2}$ to $M_i$, $i=1,3$.
The following argument shows that,  for any $0\neq [\zeta]\in  H^2_{0\,\mathrm{dR}}(M_3,\g^\ast)$,
\begin{flalign}\label{eqn:tmpextendedargument}
\PhaseSpace(f_1)\big([\underline{\mathcal{F}}^\ast_3(\zeta)]\big)= \big[\underline{\mathcal{F}}^\ast_1(\underline{f_1}_\ast(\zeta))\big]\not\in
 \big[\NN_1^0\big]~,
\end{flalign} 
where $\big[\NN_1^0\big]= \NN_1^0/\MWfunc(\Xi_1)$ is the radical of $\PhaseSpace(\Xi_1)$: 
Using Proposition \ref{propo:Nmin}
and the fact that $M_1$ has compact Cauchy surfaces (which implies
$\Omega^1_{\mathrm{tc}}(M_1,\g^\ast) = \Omega^1_{0}(M_1,\g^\ast)$), any representative $\psi$
of an element $[\psi]\in \big[\NN_1^0\big]$ has a linear part satisfying $\psi_V\in\delta\dd \Omega^1_0(M_1,\g^\ast)$.
However, for $\underline{\mathcal{F}}^\ast_1(\underline{f_1}_\ast(\zeta))$ the linear part 
is $-\delta \underline{f_1}_\ast(\zeta)$, which is not of this form if $[\underline{f_1}_\ast(\zeta)]\neq 0$ 
(this property is implied by $[\zeta]\neq0$ as $H^2_{0\,\mathrm{dR}}(f_1)$ is injective). 
As $\SubGr(\Xi_3)$ contains $[\underline{\mathcal{F}}^\ast_3(\zeta)]$, for all
$\zeta\in \Omega^2_{0,\dd}(M_3,\g^\ast)$ (see the discussion of the morphism $\PhaseSpace(f_2)$ above),
(\ref{eqn:tmpextendedargument}) implies that the image of $\SubGr(\Xi_3)$ under $\PhaseSpace(f_1)$ 
is not a subgroup of the radical $\big[\NN_1^0\big]$.
Thus, by the definition of subfunctor, any admissible choice of $\SubGr(\Xi_1)$ can not be a 
subgroup of $\big[\NN_1^0\big]$  and as a consequence the subfunctor $\SubGr$ can not be quotientable.

Let us now consider $m=2$, which is rather special and deserves for a different construction: 
Let us take the disjoint union $M_1:= \bbR\times (\mathbb{S}^1\amalg \mathbb{S}^1)$
and equip each component with the standard Lorentzian metric $g_1= -dt\otimes dt + d\phi \otimes d\phi$.
Consider further $M_2:= \bbR^2$ equipped with the Minkowski metric $g_2 = -dt\otimes dt + dx\otimes dx$.
Define $M_3$ as the Cauchy development of $\{0\} \times ( I_1\amalg I_2)\subseteq M_1$,
where $I_1$ and $I_2$ are open intervals in $(0,2\pi)$,
and equip it with the induced Lorentzian metric $g_3:= g_1\vert_{M_3}$. 
There are obvious isometric embeddings $\underline{f_1}: (M_3,g_3)\to (M_1,g_1)$
and $\underline{f_2}: (M_3,g_3)\to (M_2,g_2)$. One possibility for the latter is given by restricting to $M_3$ 
the isometric embedding $M_1 \supset \bbR \times (I_1\amalg I_2) \to  M_2$ defined by $(t,\phi_1) \mapsto  
(t,\phi_1)$ and $(t,\phi_2)\mapsto (t,2\pi + \phi_2)$, where $\phi_1\in I_1$, $\phi_2\in I_2$ and we have taken
Cartesian coordinates $(t,x)$ on $M_2= \bbR^2$.
As above, this gives rise to three objects
$\Xi_i$, $i=1,2,3$, and two morphisms
 $f_j : M_3\times G \to M_j \times G$, $j=1,2$,
 in $G{-}\PrBu^{(2)}$.
By the same arguments as in the proof of Proposition $\ref{propo:nonlocal}$,
the morphism $\PhaseSpace(f_2)$ can not be injective since $H^2_{0\,\mathrm{dR}}(M_3,\g^\ast) \simeq \bbR^2$ and
$H^2_{0\,\mathrm{dR}}(M_2,\g^\ast) \simeq \bbR$. To turn $\PhaseSpace(f_2): \PhaseSpace(\Xi_3)\to \PhaseSpace(\Xi_2)$
into an injective morphism we have to take a quotient $\PhaseSpace(\Xi_3)/\SubGr(\Xi_3)$, where
$\SubGr(\Xi_3)$ contains $[\underline{\mathcal{F}}^\ast_3(\zeta)]$, for all $\zeta\in\Omega^2_{0,\dd}(M_3,\g^\ast)$
such that $\int_{M_3}\zeta =0$.  It is easy to see that 
$H^2_{0\,\mathrm{dR}}(f_1): H^2_{0\,\mathrm{dR}}(M_3,\g^\ast)\to H^2_{0\,\mathrm{dR}}(M_1,\g^\ast)$ 
is an isomorphism, since both $M_3$ and $M_1$ have two connected components.
By the same argument as above for $m>2$ one shows that any subgroup $\SubGr(\Xi_1)$ which contains the image
of $\SubGr(\Xi_3)$ under $\PhaseSpace(f_1)$ is not a subgroup 
of the radical $\big[\NN^0_1\big]$. Hence, the subfunctor $\SubGr$ can not be quotientable.
\end{proof}


\section{\label{sec:quantization}Quantization}

We finally study the quantization of the functor $\PhaseSpace: G{-}\PrBu\to \PAG$. Due to Theorem \ref{app:CCRfunc1}
we can construct a covariant functor $\QFT : G{-}\PrBu \to \CastAlg$ by composing the functors 
$\PhaseSpace$ and $\CCR: \PAG\to \CastAlg$, 
i.e.~$\QFT := \CCR\circ \PhaseSpace$. The functor $\QFT$ describes the association of $C^\ast$-algebras
of observables $\QFT(\Xi)$ to objects $\Xi$ in $G{-}\PrBu$.
The validity of the classical causality property and of the classical time-slice axiom for the functor $\PhaseSpace$
implies the quantum causality property and the quantum time-slice axiom due to the construction
of the functor $\CCR$. Likewise, the failure of the locality property extends from the 
classical context to the quantum case, as we shall demonstrate in the remainder of this section. 
We can summarize these results as follows:
\begin{theo}
There exists a covariant functor $\QFT := \CCR\circ \PhaseSpace : G{-}\PrBu \to \CastAlg$ describing
$C^\ast$-algebras of gauge invariant on-shell observables for quantized
principal $G$-connections. The covariant functor $\QFT$ satisfies the quantum causality property and the quantum time-slice axiom.
Furthermore, for each object $\Xi$ in $G{-}\PrBu$ the $C^\ast$-algebra $\QFT(\Xi)$ is a quantization
of an algebra of functionals that separates gauge 
equivalence classes of connections (cf.~Theorem \ref{thmSep}).
Neither the functor $\QFT:G{-}\PrBu\to \CastAlg$ nor its restriction to the full subcategory
$G{-}\PrBu^{(m)}$ (for $m\geq 2$) or $G{-}\PrBu^{(m)}_0$ (for $m\geq 3$) satisfies 
the locality property, i.e.~$\QFT$ is not a covariant functor to the subcategory
$\CastAlg^\mathrm{inj}$. As an exceptional case, the functor $\QFT$ satisfies the locality property when
restricted to the full subcategory $G{-}\PrBu^{(2)}_0$ of principal $G$-bundles over 
 two-dimensional connected spacetimes.
\end{theo}

As a first step in the analysis of the quantum locality property, we 
notice that the characterization of injective morphisms $\PhaseSpace(f)$ given in Theorem
\ref{theo:diagcomap} extends to the quantized case.
\begin{theo}
Let $f: \Xi_1\to\Xi_2$ be a morphism in $G{-}\PrBu$.
 Then $\QFT(f): \QFT(\Xi_1)\to\QFT(\Xi_2)$ is injective if and only if
 $H^2_{0\,\mathrm{dR}} (f): H^2_{0\,\mathrm{dR}}(\Xi_1)\to H^2_{0\,\mathrm{dR}}(\Xi_2)$
 is injective.
\end{theo}
\begin{proof}
We show that the $\CastAlg$-morphism $\QFT(f) = \CCR\big(\PhaseSpace(f)\big)$ is injective if and only if 
the $\PAG$-morphism $\PhaseSpace(f)$ is injective, from which the proof follows by using
 Theorem \ref{theo:diagcomap}.

To prove the direction ``$\Leftarrow$'', notice that if $\PhaseSpace(f)$ is injective, 
then by Corollary \ref{cor:injectivealgmorphism} also $\QFT(f) = \CCR\big(\PhaseSpace(f)\big)$
is injective. We show the direction ``$\Rightarrow$'' by contraposition:
If $\PhaseSpace(f)$ is not injective, then there exists a $[\varphi]\in \PhaseSpace(\Xi_1)$, such that
$\PhaseSpace(f)([\varphi])=0$. As a consequence, we obtain in the CCR-algebras 
$\QFT(f) ( W_1([\varphi]) -\oone_1) = W_2\big(\PhaseSpace(f)([\varphi])\big) - \oone_2 = W_2(0)-\oone_2 =0$, 
hence $\QFT(f)$ is not injective.
\end{proof}

It remains to extend our no-go Theorem \ref{theo:noinjectivequotient} to the functor 
$\QFT: G{-}\PrBu\to\CastAlg$ and its restriction to the full subcategories 
$G{-}\PrBu^{(m)}$ (for $m\geq 2$) or $G{-}\PrBu^{(m)}_0$ (for $m\geq 3$).
This requires us to adapt (in an obvious way) the content of 
Definition \ref{defi:quotient} to functors with values in the category $\CastAlg$:
A covariant functor $\SubGr: \mathsf{C}\to \CastAlg $ is called a subfunctor 
of a covariant functor $\mathfrak{F}: \mathsf{C} \to \CastAlg$ if for all objects $A$ in $\mathsf{C}$
we have that $\SubGr(A)\subseteq \mathfrak{F}(A)$ is a $C^\ast$-subalgebra (not necessarily unital) and if
for all morphisms $f:A_1\to A_2$ in $\mathsf{C}$ the morphism $\SubGr(f)$ is the restriction
of $\mathfrak{F}(f)$ to $\SubGr(A_1)$. This subfunctor is called quotientable if 
for all objects $A$ in $\mathsf{C}$ the $C^\ast$-subalgebra $\SubGr(A)$  is a closed 
two-sided $\ast$-ideal in $\mathfrak{F}(A)$.
It is clear that under these assumptions the analog of Proposition \ref{propo:quotientphasespace} holds true, i.e.\ 
the functor $\mathfrak{F}/\SubGr: \mathsf{C}\to \CastAlg$ exists.
It is worth to  mention that, according to the definition above, a quotientable
subfunctor $\SubGr$ of $\mathfrak{F}: \mathsf{C}\to \CastAlg$ could be such that
for some object $A$ in $\mathsf{C}$ we have $\SubGr(A) = \mathfrak{F}(A)$
and hence $\mathfrak{F}(A)/\SubGr(A)\simeq \bbC$ is the trivial $C^\ast$-algebra. 
In order to avoid such trivial constructions, we shall add the requirement
that $\SubGr$ should be a proper subfunctor of $\QFT$, i.e.\ that $\SubGr(A)$ is a proper 
$C^\ast$-subalgebra of $\QFT(A)$ for all objects $A$ in $\mathsf{C}$. In this way our
no-go Theorem \ref{theo:noinjectivequotient} extends to the quantum level.
\begin{theo}\label{theo:noinjectivequotientquantum}
For any $m\geq 2$ there exists no quotientable subfunctor  $\SubGr$ of $\QFT: G{-}\PrBu^{(m)} \to \CastAlg$, such that
$\QFT/\SubGr: G{-}\PrBu^{(m)}\to \CastAlg$ satisfies the locality property. 
Furthermore, for any $m\geq 3$ there exists no quotientable subfunctor $\SubGr$ of
$\QFT : G{-}\PrBu^{(m)}_0\to \CastAlg$, such that $\QFT/\SubGr: G{-}\PrBu^{(m)}_0 \to \CastAlg$ satisfies the locality property.
\end{theo}
\begin{rem}
This theorem implies that there exists no quotientable subfunctor  $\SubGr$ of $\QFT: G{-}\PrBu \to \CastAlg$, 
such that $\QFT/\SubGr: G{-}\PrBu\to \CastAlg$ satisfies the locality property.
\end{rem}
\begin{proof}
Consider the same diagram of morphisms in $G{-}\PrBu^{(m)}$ as constructed in Theorem \ref{theo:noinjectivequotient}.
Notice that for $m\geq 3$ this is also a diagram in the full subcategory $G{-}\PrBu^{(m)}_0$.
By applying the functor $\QFT$, this gives a diagram in the category $\CastAlg$, i.e.\
\begin{flalign}
\xymatrix{
\QFT(\Xi_1) & & \QFT(\Xi_2)\\
&\ar[lu]^-{\QFT(f_1)}\QFT(\Xi_3)\ar[ru]_-{\QFT(f_2)}&
}
\end{flalign}
In the proof of Theorem \ref{theo:noinjectivequotient} we have shown that 
there exist non-trivial elements $0\neq [\zeta]\in H^2_{0\,\mathrm{dR}}(M_3,\g^\ast)$, 
such that $\PhaseSpace(f_2)\big([\underline{\mathcal{F}}_3^\ast(\zeta)]\big)=0$. As a consequence,
 $\QFT(f_2)\big(W_3\big([\underline{\mathcal{F}}_3^\ast(\zeta)]\big)\big) =
  W_2\big(\PhaseSpace(f_2)\big([\underline{\mathcal{F}}_3^\ast(\zeta)]\big)\big)= W_2(0)=\oone_2$
and thus $W_3\big([\underline{\mathcal{F}}_3^\ast(\zeta)]\big)-\oone_3$
is an element in the kernel $\ker\big(\QFT(f_2)\big)$.
In order to turn $\QFT(f_2)$ into an injective map we have to take the quotient of
$\QFT(\Xi_3)$ by a closed two-sided $\ast$-ideal $\SubGr(\Xi_3)$ which contains $\ker\big(\QFT(f_2)\big)$. 
On the other hand, as it was also shown in the proof of Theorem \ref{theo:noinjectivequotient} 
(cf.\ (\ref{eqn:tmpextendedargument})), the morphism $\PhaseSpace(f_1)$
maps all elements $[\underline{\mathcal{F}}_3^\ast(\zeta)]$, $[\zeta]\neq 0$, out 
of the radical $\big[\NN^0_1\big]$. From (\ref{eqn:N0}) it
is clear that the cohomology class $\big[h^{-1}\big(G_{(1)}(\psi_V)\big)\big]
\in H^1_\mathrm{dR}(M_1,\g)$ of any representative $\psi$ of $\PhaseSpace(f_1)\big([\underline{\mathcal{F}}_3^\ast(\zeta)]\big)$,
$[\zeta]\neq 0$, is non-vanishing. Hence, taking any $0\neq [\zeta]\in H^2_{0\,\mathrm{dR}}(M_3,\g^\ast)$
for which $[\underline{\mathcal{F}}_3^\ast(\zeta)]$ lies in the kernel of $\PhaseSpace(f_2)$,
we can find some real number $\lambda \in \bbR$ such that any representative of
$\PhaseSpace(f_1)\big([\underline{\mathcal{F}}_3^\ast(\lambda \zeta)]\big)$ violates the integral cohomology
condition in $\NN_1$ (cf.\ Proposition \ref{propo:null1}).
(Of course, $[\underline{\mathcal{F}}_3^\ast(\lambda \zeta)]$ 
is still an element in the kernel of $\PhaseSpace(f_2)$.) 
 In other words, $\PhaseSpace(f_1)\big([\underline{\mathcal{F}}_3^\ast(\lambda \zeta)]\big)\not\in \big[\NN_1\big]$
 and since $\big[\NN_1\big]$ characterizes by definition the central Weyl symbols (see (\ref{eqn:N})),
 this implies that $\QFT(f_1)$ maps $W_3\big([\underline{\mathcal{F}}_3^\ast(\lambda \zeta)]\big)$ out of the center in
  $\QFT(\Xi_1)$.
  (Of course, $\QFT(f_2)$ maps $W_3\big([\underline{\mathcal{F}}_3^\ast(\lambda \zeta)]\big)$ to the unit 
  $\oone_2$.)
 Let us denote by $W_1([\varphi])$ the image of this element 
 and by $W_1([\psi])$ a Weyl symbol that does not commute with $W_1([\varphi])$.
  For consistently 
inducing the morphism $\QFT(f_1)$ to the quotient $\QFT(\Xi_3)/\SubGr(\Xi_3)$, we have to take
the quotient $\QFT(\Xi_1)/\SubGr(\Xi_1)$ by a closed two-sided $\ast$-ideal $\SubGr(\Xi_1)$ of $\QFT(\Xi_1)$
that contains in particular the element $W_1([\varphi])-\oone_1\in \SubGr(\Xi_1)$ (as this is the image of
 $W_3\big([\underline{\mathcal{F}}_3^\ast(\lambda\zeta)]\big)-\oone_3$, lying in the kernel of $\QFT(f_2)$, under $\QFT(f_1)$). 
 Then also $W_1(-[\psi])\,\big(W_1([\varphi]) -\oone_1\big)\,W_1([\psi]) \in \SubGr(\Xi_1)$, which upon using the Weyl relations
 (\ref{eqn:Weylrelations}) simplifies to $W_1([\varphi])\,e^{-i\tau([\varphi],[\psi])}-\oone_1\in \SubGr(\Xi_1)$.
 Finally, subtracting from this the element $e^{-i\tau([\varphi],[\psi])}\,\big(W_1([\varphi])-\oone_1\big)\in\SubGr(\Xi_1)$
 we find $\big(e^{-i\tau([\varphi],[\psi])} -1\big)\,\oone_1\in\SubGr(\Xi_1)$ and hence $\oone_1\in\SubGr(\Xi_1)$,
 since $[\varphi]$ and $[\psi]$ have been chosen such that $e^{-i\tau([\varphi],[\psi])}\neq 1$.
 As a consequence, $\SubGr(\Xi_1) = \QFT(\Xi_1)$ is not a proper closed two-sided $\ast$-ideal and
thus the subfunctor $\SubGr$ would not be proper.
 \end{proof}

\begin{rem}\label{rem:topqft}
In analogy to \cite[Section 6]{Benini:2013tra}, the covariant functor $\QFT:G{-}\PrBu\to\CastAlg$ has a generally covariant
topological quantum field which measures the Chern class of the principal $G$-bundle. More precisely,
 we can construct a natural transformation $\Psi^\mathrm{mag}$ from the
singular homology functor $\mathfrak{H}_2$ to the functor $\QFT$ as follows: 
Let us denote by $\mathfrak{H}_2: G{-}\PrBu\to \mathsf{Monoid}$ the covariant functor associating
to any object $\Xi$ in $G{-}\PrBu$ the singular homology group $H_2(M,\g^\ast)$ (considered as a monoid with respect to $+$) 
of the base space. To any morphism $f:\Xi_1\to \Xi_2$ in $G{-}\PrBu$ the functor
associates the usual morphism of singular homology groups, considered as a morphism in the category $\mathsf{Monoid}$.
We further use the forgetful functor $\CastAlg \to \mathsf{Monoid}$, which forgets all structures
of $C^\ast$-algebras (but the multiplication) and turns the multiplication into a monoid structure. With a slight abuse of notation
we use the same symbol
$\QFT$ to denote the covariant functor $\QFT:G{-}\PrBu\to\mathsf{Monoid}$.
With $\mathcal{K} : H_2(M,\g^\ast)\to H^2_{0\,\mathrm{dR}^\ast}(M,\g^\ast)$ denoting the natural isomorphism
described in \cite[Section 6]{Benini:2013tra} we can define for each object $\Xi$ in $G{-}\PrBu$ a map
\begin{flalign}
\Psi^\mathrm{mag}_\Xi : \mathfrak{H}_2(\Xi) \to \QFT(\Xi)~,~~\sigma \mapsto W\big(\big[\underline{\mathcal{F}}^\ast\big(\mathcal{K}(\sigma)\big)\big]\big)~.
\end{flalign}
Notice that $\Psi^\mathrm{mag}_\Xi$ is a morphism of monoids, since, for all $\sigma,\sigma^\prime\in  \mathfrak{H}_2(\Xi) $,
\begin{flalign}
\nn W\big(\big[\underline{\mathcal{F}}^\ast\big(\mathcal{K}(\sigma +\sigma^\prime)\big)\big]\big)
&=W\big(\big[\underline{\mathcal{F}}^\ast\big(\mathcal{K}(\sigma)\big)\big]  +\big[\underline{\mathcal{F}}^\ast\big(\mathcal{K}(\sigma^\prime)\big)\big]\big) \\
&= W\big(\big[\underline{\mathcal{F}}^\ast\big(\mathcal{K}(\sigma)\big)\big]\big) 
\,W\big(\big[\underline{\mathcal{F}}^\ast\big(\mathcal{K}(\sigma^\prime)\big)\big]\big)~,
\end{flalign}
where in the last equality we have used the Weyl relation  (\ref{eqn:Weylrelations}) and the fact that
 $\tau([\underline{\mathcal{F}}^\ast(\mathcal{K}(\sigma))],
 [\underline{\mathcal{F}}^\ast(\mathcal{K}(\sigma^\prime))])=0$,
 which follows from $\underline{\mathcal{F}}^\ast(\mathcal{K}(\sigma))_V=0$.
 The collection $\Psi^\mathrm{mag} = \{\Psi^\mathrm{mag}_\Xi\}$ is then a natural transformation
 from $\mathfrak{H}_2$ to $\QFT$ that associates to elements in the second singular homology group observables that
 can measure the Chern class of the principal $G$-bundle.
\end{rem}


\section{\label{sec:locality}The locality property in Haag-Kastler-type quantum field theories}
We have shown in Theorem \ref{theo:noinjectivequotient} that there exists no quotientable subfunctor $\SubGr$
of the on-shell functor $\PhaseSpace$, such that 
the covariant functor $\PhaseSpace/\SubGr$ satisfies the locality property (unless we restrict
the functor $\PhaseSpace$ to the very special full subcategory $G{-}\PrBu^{(2)}_0$ of principal
$G$-bundles over two-dimensional connected spacetimes). 
The same result extends to the quantized level as shown in Theorem \ref{theo:noinjectivequotientquantum}.
In this section we shall prove that if we fix any object $\widehat{\Xi}=
((\widehat{M},\widehat{\o},\widehat{g},\widehat{\t}),(\widehat{P},\widehat{r}))$ of the category $G{-}\PrBu$
and consider a suitable category of subsets of $\widehat{M}$, then there exists a
quotientable subfunctor such that the resulting quotient satisfies the locality property. 
This setting does of course not cover the full generality of locally covariant quantum field theory,
however, it provides us with a quantum field theory in the sense of Haag and Kastler (generalized to curved spacetimes), 
where the focus is on associating algebras to suitable subsets of a fixed spacetime in a coherent way.

Let us fix any object $\widehat{\Xi}=
((\widehat{M},\widehat{\o},\widehat{g},\widehat{\t}),(\widehat{P},\widehat{r}))$ 
 of the category $G{-}\PrBu$.
We denote by $\SubM$ the following category: The objects in $\SubM$ are causally compatible and globally hyperbolic
open subsets of $\widehat{M}$. The morphisms in $\SubM$ are given by the subset relation $\subseteq$, 
i.e.~for any two objects $M_1,M_2$ in $\SubM$ there is a unique morphism $M_1\to M_2$ if and only if $M_1\subseteq M_2$.
Notice that by definition there exists for any object $M$ in $\SubM$ a unique morphism $M\to \widehat{M}$, 
i.e.~$\widehat{M}$ is a terminal object in $\SubM$. We interpret $\widehat{M}$ physically as the whole spacetime
(the universe).
There exists further a covariant functor $\Pull: \SubM \to G{-}\PrBu$: To any object $M$ in $\SubM$
the functor associates the object $\Pull(M)$ in $G{-}\PrBu$ obtained by pulling back all the geometric
data of $\widehat{\Xi}$  to $M$. To any morphism $M_1\to M_2$ in $\SubM$ the functor
associates the canonical embedding $\Pull(M_1)\to \Pull(M_2)$, which is a morphism
in $G{-}\PrBu$. We can compose the covariant functor $\Pull: \SubM\to G{-}\PrBu$ 
with the on-shell  functor $\PhaseSpace:= \PhaseSpaceOff/\MWfunc: G{-}\PrBu\to \PAG$  and obtain
a covariant functor  $\PS:= \PhaseSpace\circ \Pull: \SubM \to \PAG$.

We shall make heavy use of the following fact: Let $M_1\to M_2$ be any morphism
in $\SubM$, then by definition of the category $\SubM$ there exists a commutative diagram:
\begin{flalign}\label{eqn:diagramm1}
\xymatrix{
 &\widehat{M}&\\
M_1\ar[ru] \ar[rr] && M_2 \ar[lu]
}
\end{flalign}
Due to functoriality this induces the commutative diagram:
\begin{flalign}\label{eqn:diagramm2}
\xymatrix{
 &\PS(\widehat{M})&\\
\PS(M_1)\ar[ru] \ar[rr] && \PS(M_2) \ar[lu]
}
\end{flalign}
\begin{lem}
For any object $M$ in $\SubM$ define $\Kerr(M)$ to be the object in $\PAG$
given by the kernel of the canonical map $\PS(M)\to \PS(\widehat{M})$.
For any morphism $M_1\to M_2$ in $\SubM$ define the morphism $\Kerr(M_1)\to \Kerr(M_2)$
in $\PAG$ by restriction of $\PS(M_1)\to \PS(M_2)$ to $\Kerr(M_1)$.
Then $\Kerr: \SubM \to \PAG$ is a quotientable subfunctor of $\PS:\SubM\to \PAG$.
\end{lem}
\begin{proof}
For any object $M$ in $\SubM$, the kernel $\Kerr(M)$ of the canonical map $\PS(M)\to \PS(\widehat{M})$
is an object in $\PAG$ with the presymplectic structure induced from $\PS(M)$, that becomes trivial
in $\Kerr(M)$. Furthermore, due to the commutative diagram (\ref{eqn:diagramm2}),
the restriction of any morphism $\PS(M_1)\to \PS(M_2)$ to $\Kerr(M_1)$
induces a morphism $\Kerr(M_1)\to \Kerr(M_2)$ in $\PAG$. The composition property of $\PS$
is preserved and hence $\Kerr: \SubM \to \PAG$ is a quotientable subfunctor
of $\PS:\SubM\to \PAG$.
\end{proof}
We can now prove the main statement of this section.
\begin{theo}\label{thm:HKworks}
The covariant functor $\PS^0 := \PS / \Kerr: \SubM \to \PAG$ satisfies the locality property,
the causality property and the time-slice axiom.
\end{theo}
\begin{proof}
The causality property and the time-slice axiom are induced since $\PhaseSpace: G{-}\PrBu\to \PAG$ satisfies these 
properties. To prove the locality property we have to check if all morphism
$\PS^0(M_1)\to \PS^0(M_2)$ are injective. Since by definition $\PS^0(M_i) = \PS(M_i)/\Kerr(M_i)$, for $i=1,2$,
and due to the commutative diagram (\ref{eqn:diagramm2}) we obtain the commutative diagram
\begin{flalign}
\xymatrix{
 &\PS^0(\widehat{M})= \PS(\widehat{M})&\\
\PS^0(M_1)\ar[ru] \ar[rr] && \PS^0(M_2) \ar[lu]
}
\end{flalign}
where both upwards going arrows are injective (by construction). As a consequence, the horizontal arrow has to be
injective too, which proves the locality property.
\end{proof}

As a consequence of this theorem we can consider $\PS^0$ as a covariant functor 
$\PS^0 : \SubM \to \PAG^\mathrm{inj}$, where the latter category is defined in Definition \ref{defi:PAG}.
Using further Theorem \ref{app:CCRfunc2} we obtain a covariant functor
$\QFT_{\widehat{\Xi}}^0 := \CCR\circ \PS^0 : \SubM \to \CastAlg^\mathrm{inj}$.
\begin{cor}
The  covariant functor
$\QFT_{\widehat{\Xi}}^0 : \SubM \to \CastAlg^\mathrm{inj}$
satisfies the quantum
causality property, the quantum time-slice axiom and the locality property.
\end{cor} 
\begin{rem}\label{rem:HKworks}
Following the ideas presented in \cite[Proposition 2.3]{Brunetti:2001dx} we can construct a theory in the sense of
 Haag and Kastler from the covariant functor $\QFT_{\widehat{\Xi}}^0:\SubM\to \CastAlg^\mathrm{inj}$:
Let us consider the set $\mathrm{Obj}(\SubM)$ of objects in $\SubM$, i.e.~the set of all causally compatible and globally hyperbolic
open subsets of the reference spacetime $\widehat{M}$. 
The functor $\QFT_{\widehat{\Xi}}^0$ associates to the terminal object $\widehat{M}$ a $C^\ast$-algebra 
$\QFT_{\widehat{\Xi}}^0(\widehat{M})$, which we shall interpret as the global algebra of observables.
To any element $M\in \mathrm{Obj}(\SubM)$ the functor associates a $C^\ast$-algebra $\QFT_{\widehat{\Xi}}^0(M)$,
which can be mapped with an injective unital $C^\ast$-algebra homomorphism into $\QFT_{\widehat{\Xi}}^0(\widehat{M})$.
With a slight abuse of notation we denote the image of $\QFT_{\widehat{\Xi}}^0(M)$ under this map by the same symbol.
Hence, we have an association
\begin{flalign}\label{eqn:association}
 \mathrm{Obj}(\SubM) \ni M \mapsto \QFT_{\widehat{\Xi}}^0(M)\subseteq \QFT_{\widehat{\Xi}}^0(\widehat{M})~.
\end{flalign}
Following the proof of \cite[Proposition 2.3]{Brunetti:2001dx} one can show that this association
satisfies isotony, causality and the time-slice axiom. Furthermore,
if there is a group of orientation and time-orientation preserving isometries
acting on $\widehat{M}$, the association (\ref{eqn:association}) is covariant.
Hence, it is a quantum field theory in the sense of Haag and Kastler
\cite{Haag:1963dh}, generalized to an arbitrary but fixed spacetime $\widehat{M}$.
\end{rem}

The considerations in this section make heavy use of a terminal object in the category
$\SubM$. Indeed, the existence of this object has provided us with commutative diagrams
of the form (\ref{eqn:diagramm1}), which are essential for constructing a suitable quotientable subfunctor.
With this subfunctor we could construct quantum field theories in the sense of Haag and Kastler.
Notice that the category $G{-}\PrBu$ has no terminal object,
hence the techniques developed in this section do not apply to this case.
This is of course already clear from Theorem \ref{theo:noinjectivequotient} and 
Theorem \ref{theo:noinjectivequotientquantum}, where
it is shown that there exists no quotientable subfunctor which leads to a theory obeying the strict
axioms of general local covariance \cite{Brunetti:2001dx}.

\section{\label{sec:physical}Concluding physical remarks}
The theory of electromagnetism contains several features which are connected to the topology of 
a region $M$ of an $m$-dimensional spacetime $\widehat{M}$ in an algebraic description -- the Aharonov-Bohm effect, 
related to $H^1_{\mathrm{dR}}(M)$, as well as electric and magnetic charges, related to $H^{m-2}_\mathrm{dR}(M)$
and $H^2_{\mathrm{dR}}(M)$, respectively. 
Describing electromagnetism as a theory of principal $U(1)$-connections in its entirety, we have been able to 
provide a quantum framework which describes all of these topological features 
in a coherent manner. In the following we briefly comment on the relation of our constructions 
and results to these physical aspects of electromagnetism.

To discuss the Aharonov-Bohm effect, we recall the main aspects of \cite[Example 3.1]{SDH12} 
(see also \cite{Leyland:1978iv} for an early account of the Aharonov-Bohm effect in the algebraic framework): 
Consider as a globally hyperbolic spacetime $M$ the Cauchy development in $4$-dimensional Minkowski spacetime $\widehat{M}$
 of the time-zero hypersurface $\{0\}\times\mathbb{R}^3$ with the $z$-axis removed.
The (necessarily trivial) principal $U(1)$-bundle over $\widehat{M}$ pulls back to a
trivial principal $U(1)$-bundle over $M$. One has $H^1_{\mathrm{dR}}(M,\g)\neq\{0\}$ and, choosing the trivial 
connection as a reference, $H^1_{\mathrm{dR}}(M,\g)$ can be spanned by the on-shell 
vector potential $i\,d \phi \in \Omega^1(M,\g)$, with $\phi$ being the azimuthal angle around the $z$-axis
 in cylindrical coordinates $(t,\rho,\phi,z)$ on $M$. Here the $z$-axis represents an infinitely 
 thin coil whose magnetic flux $\Phi$ through the plane perpendicular to the coil can be encoded 
 in the vector potential $i\,\frac{\Phi}{2\pi} d \phi$. The gauge invariant affine functionals 
 (\ref{eqn:affobs}) introduced in \cite{Benini:2013tra} can not distinguish connections with the different gauge potentials
  $i\,\frac{\Phi}{2\pi} d \phi$, $\Phi\in\bbR$, and are thus not sufficient to measure this flux,
   cf.~\cite[Remark 4.5]{Benini:2013tra}. 
The exponential observables (\ref{eqn:expobs}) solve this problem and further 
 shortcomings in previous treatments of the subject. On the one hand, they contain Aharonov-Bohm observables
  which in contrast to the ones in \cite{Dimock, Pfenning:2009nx, SDH12, Finster:2013fva} are fully gauge invariant 
  and measure the phase $\exp i \Phi$ rather  than the flux $\Phi$ itself.  This is consistent with and, indeed, reproduces the
  Aharonov-Bohm experiment. On the other hand, they are regular enough for quantization, in contrast 
  to Aharonov-Bohm observables of Wilson-loop type. In fact, they can be considered as regularized Wilson loops,
   cf.~Remark \ref{rem:wilson}.

In \cite{DL} it has been found that the sensitivity of electromagnetism to $H^2_{\mathrm{dR}}(M)$ 
(and $H^{m-2}_\mathrm{dR}(M)$) leads to a failure of the locality axiom in locally covariant quantum field theory 
as introduced in \cite{Brunetti:2001dx}. This has been confirmed in \cite{SDH12, Benini:2013tra} and 
in Proposition \ref{propo:nonlocal} of this work and ascribed to the Gauss law in \cite{SDH12}. 
To understand this in view of our results, we introduce a few novel notions. 
\begin{defi}
An {\bf electric charge observable} is a gauge invariant exponential functional
$\mathcal{W}_\varphi$, with $ \varphi= \underline{\mathcal{F}}^\ast(\zeta)$, 
$\zeta \in \Omega^2_{0,\dd}(M,\g^\ast)$ and $0\neq [\zeta] \in H^2_{0\,\mathrm{dR}}(M,\g^\ast)$. 
An {\bf electrically charged configuration} is an on-shell connection, i.e.~$\lambda\in \sect{M}{\mathcal{C}(\Xi)}$
and $\MW(\lambda)=0$, such that there exists an electric charge observable 
with $\mathcal{W}_\varphi (\lambda)\neq 1$. 
In the notation of Section \ref{sec:locality}, for any object $M$ in $\SubM$ a {\bf material electric charge observable} 
is an electric charge observable $\mathcal{W}_\varphi$, such that $0\neq [\varphi]\in\Kerr(M)$.
A {\bf materially electric charged configuration} is an on-shell connection $\lambda$,
such that for some material electric charge observable $\mathcal{W}_\varphi (\lambda)\neq 1$.
\end{defi}

It is easy to prove that no materially electric charged configuration on $M$ can be extended 
from $M$ to $\widehat{M}$.  In an interacting theory including charged matter fields, these 
configurations can be interpreted as the 
 connections sourced by an electric current density located in $\widehat{M}\setminus M$, 
 cf.~\cite[Example 3.7]{SDH12}. Thus we conclude that materially electric charged configurations
  are unphysical in pure electromagnetism and have to be discarded by considering 
  an appropriate subset of the solution space.
 Consequently, by duality, all 
material electric charge observables have to be discarded by taking an appropriate quotient as in Theorem \ref{thm:HKworks}, 
  since, in the absence of materially electric charged configurations, they always 
  give a vanishing measurement result. One may consider singling out even more observables, 
  or dually, taking even smaller subsets of configurations, but there is no apparent reason for 
  doing so, as all remaining observables and configurations make perfect sense in pure electromagnetism, cf.\ also the next paragraph.
  Thus, in view of Remark \ref{rem:HKworks} one may say that {\it the  quotient taken 
  in Theorem \ref{thm:HKworks} gives for each region $M$ of an arbitrary
  but fixed spacetime $\widehat{M}$ the correct, full algebra of observables of this region in pure electromagnetism.} 
  In this respect Theorem \ref{theo:noinjectivequotientquantum} can be interpreted as to imply that
  {\it it is mathematically impossible to construct the correct algebra of observables of a region of 
  spacetime in pure electromagnetism without knowing a priori the whole spacetime.} 
 In the counterexample constructed in the proof of this theorem, considering three objects $\Xi_i$, $i=1,2,3$,
and two morphisms $f_j : \Xi_3\to \Xi_j$, $j=1,2$, in $G{-}\PrBu$, the electric charge observables
 on $M_3$ are material with respect to $M_2$ but not with respect to $M_1$. Thus, they might or might 
 not belong to the correct algebra of observables depending on whether $M_1$ or $M_2$ is the whole spacetime. 
 We expect that this problem disappears in an interacting theory containing also dynamical matter field currents, 
 see also \cite[Remark 4.15]{SDH12}.

By starting from a smaller algebra which does not contain Aharonov-Bohm observables,
it is possible to construct a quantum field theory functor on the full category $G{-}\PrBu$
that satisfies the locality property, cf.~\cite[Theorem 7.3]{Benini:2013tra}.
In the context of the counterexample constructed in the proof of Theorem \ref{theo:noinjectivequotient}
 and Theorem \ref{theo:noinjectivequotientquantum}, this implies to 
discard all electric charge observables on $M_3$, even if they might be indispensable observables in pure
electromagnetism depending on the nature of the whole spacetime. Hence, the result in \cite[Theorem 7.3]{Benini:2013tra}
seems mathematically very pleasing, but it is not satisfactory from the physical point of view.


\section*{Acknowledgements}
We would like to thank Hanno Gottschalk for comments on this work and especially for his help 
in preparing the appendix. We are also grateful to Klaus Fredenhagen, Daniel Paulino and 
Katarzyna Rejzner for useful discussions. 
Furthermore, we acknowledge the helpful comments of the referees,
 which have led to substantial improvements of the manuscript.
The work of C.D.\ has been supported partly 
by the University of Pavia and partly by the Indam-GNFM project {``Influenza della materia 
quantistica sulle fluttuazioni gravitazionali''}. The work of M.B.\ has been supported partly by a 
DAAD scholarship. M.B.\ is grateful to the II.\ Institute for Theoretical Physics of the University 
of Hamburg for the kind hospitality. The work of T.-P. H. is supported by a research fellowship 
of the Deutsche Forschungsgemeinschaft (DFG).


\appendix

\section{\label{sec:weyl}CCR-representations of generic presymplectic Abelian groups}
In this appendix we discuss the generalization of the theory of Weyl systems and $\CCR$-representations
from symplectic vector spaces to generic presymplectic Abelian groups.
For the case of symplectic vector spaces the theory of $\CCR$-representations is well understood
and details can be found in \cite{Bar:2007zz,BarCCR,Bratteli:1996xq}. The generalization to presymplectic
vector spaces has been studied in \cite{degenerateCCR} and $\CCR$-representations of presymplectic Abelian
groups appeared in \cite{Manuceau:1973yn}.  We will first review the results of Manuceau et al.~\cite{Manuceau:1973yn}
and afterwards provide some further constructions which are essential
for locally covariant quantum field theory. 

 \begin{defi}\label{defi:PAG}
 \begin{itemize}
\item[(i)] A {\bf presymplectic Abelian group} is a tuple $(B,\tau)$, where $B$ is an Abelian group\footnote{
 We denote the group operation
 by $+$, the identity element by $0$ and the inverse of $b\in B$ by $-b$}
  and
 $\tau: B\times B \to \bbR$ is an antisymmetric map, such that
 $\tau(b,\,\cdot\,) : B\to \bbR\,,~b^\prime \mapsto \tau(b,b^\prime)$ is a homomorphism of Abelian groups, for all $b\in B$.
 \item[(ii)] The category $\PAG$ consists of the following objects and morphisms: 
 An object is a presymplectic Abelian group $(B,\tau)$. A morphism 
  $\phi: (B_1,\tau_1)\to (B_2,\tau_2)$ is a group homomorphism (not necessarily injective) that preserves the presymplectic structures, 
 i.e.~$\tau_2\circ (\phi \times \phi) = \tau_1$.
 \item[(iii)] The category $\PAG^\mathrm{inj}$ is the subcategory of $\PAG$ where all morphisms are injective.
\end{itemize}
 \end{defi}

To any object $(B,\tau)$ in  $\PAG$ we can associate a unital $\ast$-algebra over $\bbC$ as follows:
Consider the $\bbC$-vector space $\Delta(B,\tau)$ that is spanned by a basis $W(b)$, $b\in B$.
Any element in $\Delta(B,\tau)$ is of the form $a=\sum_{i=1}^N \alpha_i\,W(b_i)$, where
$\alpha_i\in \bbC$ and $b_i\in B$, for all $i=1,\dots,N$. We can assume without loss of generality 
 that all $b_i$'s are different in expressions like this one. We define the structure of an associative unital 
algebra on $\Delta(B,\tau)$ by, for all $b,c\in B$,
\begin{flalign}\label{eqn:Weylrelations}
W(b)\,W(c) := e^{-i\,\tau(b,c)/2}\,W(b+c)~. 
\end{flalign}
Notice that $W(0)=\oone$ is the unit element. We further define a $\ast$-structure on $\Delta(B,\tau)$ by
$W(b)^\ast := W(-b)$, for all $b\in B$, and we notice that this turns $\Delta(B,\tau)$ into a unital 
$\ast$-algebra over $\bbC$. All $W(b)$ are unitary.

Given a morphism $\phi: (B_1,\tau_1)\to (B_2,\tau_2)$ in $\PAG$
we can construct a unital $\ast$-algebra homomorphism between $\Delta(B_1,\tau_1)$ and $\Delta(B_2,\tau_2)$
as follows: For any element $a=\sum_{i=1}^N \alpha_i\,W_1(b_i) \in \Delta(B_1,\tau_1)$ we define
$\Delta(\phi) (a) := \sum_{i=1}^N \alpha_i\,W_2(\phi(b_i))$. Then
$\Delta(\phi) : \Delta(B_1,\tau_1) \to \Delta(B_2,\tau_2)$ is clearly a $\bbC$-linear map and also a
 unital $\ast$-algebra homomorphism, 
since $\phi$ is a group homomorphism preserving the presymplectic structures.
Notice that for the identity morphism
$\id_{(B,\tau)}$ in $\PAG$ we have that $\Delta(\id_{(B,\tau)}) = \id_{\Delta(B,\tau)}$. Furthermore,
given two composable morphisms $\phi_1:(B_1,\tau_1)\to (B_2,\tau_2)$
and $\phi_2: (B_2,\tau_2)\to (B_3,\tau_3)$ in $\PAG$ it is easy to check that
$\Delta(\phi_2\circ \phi_1) = \Delta(\phi_2)\circ \Delta(\phi_1)$.
Hence, $\Delta: \PAG \to \astAlg$ is a covariant functor, where the category $\astAlg$ consists
of unital $\ast$-algebras over $\bbC$ as objects and unital $\ast$-algebra homomorphisms
(not necessarily injective) as morphisms. It is easy to see that $\Delta$ restricts to a covariant
functor $\Delta :\PAG^\mathrm{inj} \to \astAlg^\mathrm{inj}$, where $\astAlg^\mathrm{inj}$
is the subcategory of $\astAlg$ where all morphisms are injective.

For constructing a suitable $C^\ast$-completion of $\Delta(B,\tau)$ we follow
the strategy of \cite{Manuceau:1973yn} and introduce as an intermediate step a $\ast$-Banach algebra.
Let us consider the $\ast$-norm $\Vert \cdot \Vert^\mathrm{Ban} : \Delta(B,\tau) \to \bbR^+$ defined by
\begin{flalign}
\Big\Vert \sum_{i=1}^N\alpha_i\,W(b_i)\Big\Vert^\mathrm{Ban} := \sum_{i=1}^N \vert\alpha_i\vert~.
\end{flalign}
We denote the completion of $\Delta(B,\tau)$ by $\Delta^\mathrm{Ban}(B,\tau)$ and notice that it is a unital $\ast$-Banach algebra.
A generic element in $\Delta^\mathrm{Ban}(B,\tau)$ is of the form $a=\sum_{i=1}^\infty\alpha_i\,W(b_i)$, with
$\alpha_i\in \bbC$ and $b_i\in B$, such that $\sum_{i=1}^\infty \vert\alpha_i\vert <\infty$.

Given a morphism $\phi: (B_1,\tau_1)\to (B_2,\tau_2)$ in $\PAG$
we note that $\Delta(\phi) : \Delta(B_1,\tau_1) \to \Delta(B_2,\tau_2)$ is bounded by $1$, 
i.e.~$\Vert \Delta(\phi)(a) \Vert^\mathrm{Ban}_2 \leq \Vert a \Vert^\mathrm{Ban}_1$, for all $a\in \Delta(B_1,\tau_1)$.
Hence, there exists a unique continuous extension of $\Delta(\phi)$ to the completions, which we denote
 by the symbol $\Delta^\mathrm{Ban}(\phi) : \Delta^\mathrm{Ban}(B_1,\tau_1)
 \to \Delta^\mathrm{Ban}(B_2,\tau_2)$. 
If the morphism $\phi$ is in $\PAG^\mathrm{inj}$, then $\Delta(\phi)$ is an isometry,
i.e.~$\Vert \Delta(\phi)(a) \Vert^\mathrm{Ban}_2 = \Vert a \Vert^\mathrm{Ban}_1$ for all $a\in \Delta(B_1,\tau_1)$. 
 In this case $\Delta^\mathrm{Ban}(\phi) $ is an isometry and hence in particular injective. 
 Furthermore, given  two composable morphisms $\phi_1:(B_1,\tau_1)\to (B_2,\tau_2)$
and $\phi_2: (B_2,\tau_2)\to (B_3,\tau_3)$ in $\PAG$ it is easy to check that
$\Delta^\mathrm{Ban}(\phi_2\circ \phi_1) = \Delta^{\mathrm{Ban}}(\phi_2)\circ \Delta^\mathrm{Ban}(\phi_1)$.
Hence, $\Delta^\mathrm{Ban}: \PAG \to \BastAlg$ is a covariant functor, where the category $\BastAlg$ consists
of unital $\ast$-Banach algebras as objects and unital $\ast$-Banach algebra homomorphisms (not necessarily injective)
as morphisms. 
 Notice that $\Delta^\mathrm{Ban}$ restricts to a covariant
functor $\Delta^\mathrm{Ban} :\PAG^\mathrm{inj} \to \BastAlg^\mathrm{inj}$, where $\BastAlg^\mathrm{inj}$
is the subcategory of $\BastAlg$ where all morphisms are injective.

In the following we shall require states on the $\ast$-Banach algebras $\Delta^{\mathrm{Ban}}(B,\tau)$,
i.e.~continuous positive linear functionals $\omega : \Delta^{\mathrm{Ban}}(B,\tau) \to \bbC$ satisfying
$\omega(\oone)=1$.
The following proposition, which is proven in \cite[Proposition (2.17)]{Manuceau:1973yn},
will be very helpful in constructing such states:
\begin{propo}\label{propo:extensionstate}
Any positive linear functional on $\Delta(B,\tau)$ extends to a continuous positive linear functional
 on $\Delta^{\mathrm{Ban}}(B,\tau)$.
\end{propo}
There  exists a faithful state on $\Delta^\mathrm{Ban}(B,\tau)$, which can be seen as follows:
 Let us define a positive linear functional
$\omega: \Delta(B,\tau) \to \bbC$ by $\omega(W(b))=0$, if $b\neq 0$, and $\omega(W(0)) = \omega(\oone) =1$.
By Proposition \ref{propo:extensionstate} we can extend $\omega$ to a continuous positive linear functional
on $\Delta^\mathrm{Ban}(B,\tau)$ (denoted by the same symbol), which satisfies
$\omega(\oone)=1$, hence it is a state. This state is faithful, i.e.~$\omega(a^\ast a) >0$ 
for any $a\in \Delta^\mathrm{Ban}(B,\tau)$, $a\neq 0$. The existence of a faithful state allows us to define
the following $C^\ast$-norm on $\Delta^\mathrm{Ban}(B,\tau)$.
\begin{defi}\label{defi:CCRalg}
Let $\mathcal{F}$ be the set of states on $\Delta^\mathrm{Ban}(B,\tau)$. 
The {\bf minimal regular norm} on $\Delta^\mathrm{Ban}(B,\tau)$ is defined by, for all $a\in \Delta^\mathrm{Ban}(B,\tau)$,
\begin{flalign}
\Vert a\Vert := \sup_{\omega\in \mathcal{F}}\sqrt{\omega(a^\ast\,a)}~.
\end{flalign}
The completion of $\Delta^\mathrm{Ban}(B,\tau)$ (or equivalently $\Delta(B,\tau)$) with respect to the minimal regular norm
is denoted by $\CCR(B,\tau)$. Then $\CCR(B,\tau)$ is a unital $C^\ast$-algebra (cf.~\cite{Manuceau:1973yn}).
\end{defi}

\begin{propo}\label{propo:morext}
Let  $\phi:(B_1,\tau_2)\to (B_2,\tau_2)$ be a morphism in
$\PAG$. Then there exists a unique continuous extension $\CCR(\phi) : \CCR(B_1,\tau_1)\to \CCR(B_2,\tau_2)$
of $\Delta^\mathrm{Ban}(\phi)$ (and hence also of $\Delta(\phi)$).
\end{propo}
\begin{proof}
We have to prove that there exists $C\in \bbR$, such that
 $\Vert \Delta^\mathrm{Ban}(\phi) (a)\Vert_2 \leq C\,\Vert a\Vert_1$, for all $a\in \Delta^\mathrm{Ban}(B_1,\tau_1)$.
 The existence and uniqueness of a continuous extension then follows by standard extension theorems.
We obtain by a straightforward calculation
 \begin{flalign}
 \Vert \Delta^\mathrm{Ban}(\phi) (a)\Vert_2  
 =\sup_{\omega\in \mathcal{F}_2} \sqrt{\omega\big(\Delta^\mathrm{Ban}(\phi)(a^\ast\,a)\big)}
 \leq \sup_{\omega^\prime \in \mathcal{F}_1} \sqrt{\omega^\prime(a^\ast\,a)} = \Vert a\Vert_1~,
 \end{flalign}
 where in the second step we have used that $\omega \circ \Delta^\mathrm{Ban}(\phi) \in \mathcal{F}_1$. Hence, $C=1$.
\end{proof}
Let us denote by $\CastAlg$ the category whose objects are unital $C^\ast$-algebras
and whose morphisms are unital $C^\ast$-algebra homomorphisms (not necessarily injective). 
The first main result of this appendix is summarized in the following
\begin{theo}\label{app:CCRfunc1}
$\CCR : \PAG \to \CastAlg$ is a covariant functor.
\end{theo}

It remains to show that $\CCR$ restricts to a covariant functor
$\CCR: \PAG^\mathrm{inj}\to \CastAlg^\mathrm{inj}$, where
$\CastAlg^\mathrm{inj}$ is the subcategory of $\CastAlg$ where all morphisms are injective.
Notice that for a morphism $\phi:(B_1,\tau_1)\to (B_2,\tau_2)$ in $\PAG^\mathrm{inj}$ the morphism
 $\CCR(\phi) : \CCR(B_1,\tau_1)\to \CCR(B_2,\tau_2)$ would be an isometry (in particular injective)
if we could prove that for any $\omega^\prime \in \mathcal{F}_1$ there exists
a $\omega\in \mathcal{F}_2$, such that $\omega\circ \Delta^\mathrm{Ban}(\phi) = \omega^\prime$.
Due to Proposition \ref{propo:extensionstate} it is sufficient to prove that for any 
normalized positive linear functional $\omega^\prime$ on $\Delta(B_1,\tau_1)$ there
 exists a normalized positive linear functional $\omega$
on $\Delta(B_2,\tau_2)$, such that $\omega\circ \Delta(\phi) = \omega^\prime$. On the image
$\Delta(\phi)[\Delta(B_1,\tau_1)]\subseteq \Delta(B_2,\tau_2)$ we can invert $\Delta(\phi)$ since it is injective
and hence arrive at the following extension problem: Does there exist a positive linear functional
$\omega:\Delta(B_2,\tau_2)\to \bbC$ extending
 $\omega^\prime\circ \Delta(\phi)^{-1}: \Delta(\phi)[\Delta(B_1,\tau_1)] \to \bbC$?
 Indeed, such an extension can be found by applying the positive-cone version of the Hahn-Banach Theorem,
 see e.g.~\cite[Theorem 2.6.2]{Edwards}.
 \begin{propo}
 Let $\phi:(B_1,\tau_1)\to (B_2,\tau_2)$ be a morphism in $\PAG^\mathrm{inj}$.
 Then there exists for any positive linear functional $\tilde{\omega} : \Delta(\phi)[\Delta(B_1,\tau_1)] \to \bbC$ 
 an extension $\omega: \Delta(B_2,\tau_2)\to \bbC$ that is a positive linear functional on $\Delta(B_2,\tau_2)$.
 \end{propo}
 \begin{proof}
 Let us denote by $H:=\{a\in \Delta(B_2,\tau_2) : a^\ast =a\} $ and $\tilde{H} := 
 \{a\in \Delta(\phi)[\Delta(B_1,\tau_1)]  : a^\ast =a\}$ the $\bbR$-vector spaces of hermitian elements.
 Notice that $\oone_2\in \tilde{H}\subseteq H$. The given positive linear functional $\tilde\omega$
 restricts to a positive $\bbR$-linear functional (denoted by the same symbol) $\tilde \omega:\tilde H\to \bbR$.
 By \cite[Theorem 2.6.2]{Edwards} we can extend $\tilde\omega$ to a positive linear functional
 $\omega: H\to \bbR$, provided that for each element $h\in H$ there exists at least one $\tilde h\in \tilde H$, such
 that $\tilde h -h$ is in the positive cone $K$ \footnote{
The positive cone here is the subset $K\subset H$ consisting of finite sums of elements
 $\beta\,a^\ast a$, with $\beta >0$ and $a\in \Delta(B_2,\tau_2)$.
}. This condition is satisfied for the following reason: Any $h\in H$ can be expressed
 as a finite sum of the basic hermitian elements $h_{\alpha,b}:= 
 \alpha \,W_2(b) + \overline{\alpha}\,W_2(-b)$, with $\alpha\in \bbC$, $b\in B_2$ and $\overline{\cdot}$ 
 denotes complex conjugation.
Hence, it is sufficient to prove that for any $\alpha\in \bbC$ and $b\in B$ there exists $\tilde{h}\in \tilde H$, 
such that $\tilde h - h_{\alpha,b} \in K$. Defining $a:= \oone_2 - \alpha\, W_2(b)$ we find
$a^\ast a = (1+\overline{\alpha}\alpha)\,\oone_2 - h_{\alpha,b}$ and thus $\tilde h - h_{\alpha,b} \in K$
for $\tilde h = (1 + \overline{\alpha}\alpha)\,\oone_2 \in\tilde{H}$.

The positive linear functional $\omega: H\to \bbR $ which is obtained by this extension procedure is further extended 
to $\Delta(B_2,\tau_2)$ as follows: For any $a\in \Delta(B_2,\tau_2)$ we define the real and imaginary part
by $a_R := (a + a^\ast)/2$ and $a_I := (a - a^\ast)/2i$. Notice that $a_R,a_I\in H $. We then extend $\omega$
to a $\bbC$-linear map on all of $\Delta(B_2,\tau_2)$ by defining $\omega(a) := \omega(a_R) + i\, \omega(a_I)$.
It is easy to see that this is an extension of $\tilde\omega$, which completes the proof.
 \end{proof}
 
\begin{cor}\label{cor:injectivealgmorphism}
For any morphism $\phi: (B_1,\tau_1)\to (B_2,\tau_2)$ in $\PAG^\mathrm{inj}$ the map
$\CCR(\phi): \CCR(B_1,\tau_1)\to \CCR(B_2,\tau_2)$ of Proposition \ref{propo:morext} is an isometry.
In particular, $\CCR(\phi)$ is an injective unital $C^\ast$-algebra homomorphism.
\end{cor}

The second main result of this appendix is summarized in the following
\begin{theo}\label{app:CCRfunc2}
$\CCR : \PAG^\mathrm{inj} \to \CastAlg^\mathrm{inj}$ is a covariant functor.
\end{theo}


\end{document}